\documentclass[12pt]{amsart}
\usepackage[margin=1in]{geometry}
\usepackage{amsmath,amssymb,amsthm,mathtools}
\usepackage{booktabs,tabularx}
\usepackage[colorlinks=true,linkcolor=blue]{hyperref}
\usepackage{enumitem}
\usepackage{placeins}
\usepackage{microtype}

\newtheorem{theorem}{Theorem}
\newtheorem{proposition}[theorem]{Proposition}
\newtheorem{lemma}[theorem]{Lemma}
\newtheorem{remark}[theorem]{Remark}
\newcommand{\R}{\mathbb R}
\newcommand{\Z}{\mathbb Z}
\newcommand{\C}{\mathbb C}
\newcommand{\T}{\mathbb T}
\newcommand{\cH}{\mathcal H}
\newcommand{\ip}[2]{\langle #1,#2\rangle}
\newcommand{\norm}[1]{\lVert #1\rVert}
\newcommand{\la}{\langle}
\newcommand{\ra}{\rangle}

\title[Sharp decay thresholds]{Sharp decay thresholds for eigenvalues of discrete Schr\"odinger operators on $\Z^2$}
\author[S.\ Ganguly]{Shirshendu Ganguly}
\address{[S.\ Ganguly] Department of Statistics, UC Berkeley}
\email{\href{mailto:sganguly@berkeley.edu}{sganguly@berkeley.edu}}

\author[W.\ Liu]{Wencai Liu}
\address{[W.\ Liu] Department of Mathematics, Texas A\&M University}
\email{\href{mailto:wencail@tamu.edu}{wencail@tamu.edu}}
\date{}

\begin{document}
	\maketitle
	\begin{abstract}
		We study eigenvalues of discrete Schr\"odinger operators
		$H=\Delta+V$ on $\Z^2$, where $\Delta$ is the uncentered Laplacian, i.e., the un-normalized adjacency operator of $\Z^2$ and $V$ decays at infinity.  By Weyl's
		theorem, the essential spectrum of $H$ is $[-4,4]$.  We determine the
		sharp decay thresholds for existence of  eigenvalues in three distinct spectral regimes. While it is natural to expect different behavior at the spectral edge
		$\lambda=\pm4$, and the bulk, there is a further distinction between the regular
		energies $0<|\lambda|<4$ and the interior critical energy $\lambda=0$ stemming from the reducibility of the corresponding Fermi surface in the latter case. For every $0<|\lambda|<4$, we construct  potentials $V$ satisfying
		$|V(n)|\leq C|n|^{-1}$ for which $\lambda$ is an eigenvalue of
		$\Delta+V$, and prove absence of eigenvalues when
		$|V(n)|\leq C|n|^{-1-\varepsilon}$ for some $\varepsilon>0$. At $\lambda=0$, the critical power changes and we construct potentials $V$
		satisfying $|V(n)|\leq C|n| ^{-2}$ for which $0$ is an eigenvalue
		of $\Delta+V$, as well as prove absence  when
		$|V(n)|\leq C|n|^{-2-\varepsilon}$ for any $\varepsilon>0$.
		Finally, at each spectral edge, we show that, for every $K\geq3$, an eigenvalue
		can be created by  potentials supported on exactly
		$K$ sites, whereas a potential supported on at most two sites cannot
		create an edge eigenvalue.  The proofs combine Green-function expansions
		and moment cancellation, Hilbert-space-valued iterations,
		discrete Carleman estimates, and a uniform Green-kernel estimate.
	\end{abstract}

	\tableofcontents
	\section{Introduction}\label{sec:introduction}
	
	Let $d\geq2$.  On $\ell^2(\Z^d)$ {we consider the nearest-neighbor
		discrete adjacency operator
		\begin{equation}\label{eq:free-laplacian-Zd}
		(\Delta u)(n)
		=\sum_{j=1}^d\bigl(u(n+e_j)+u(n-e_j)\bigr),
		\qquad n\in\Z^d,
		\end{equation}
		where $e_1,\ldots,e_d$ are the standard coordinate vectors. Letting $\T:=\R/2\pi \Z,$ under the
		Fourier transform $\mathcal F:\ell^2(\Z^d)\to L^2(\T^d)$, $\Delta$ becomes the Fourier multiplier by the factor
		\[
		h_0(\xi)=2\sum_{j=1}^d\cos\xi_j,
		\qquad \xi\in\T^d,
		\]
		and hence $\sigma(\Delta)=[-2d,2d]$}.  We study the Schr\"odinger operator
	\begin{equation}\label{eq:discrete-schrodinger}
	H=\Delta+V,
	\end{equation}
	where $V:\Z^d\to\C$ decays at infinity.  
	 
	Already for the {free operator,} {meaning just the adjacency operator,} dimension two contains three geometrically
	different spectral regimes as evident from the discussion below. Set
	\[
	p_\lambda(\xi)=2\cos\xi_1+2\cos\xi_2-\lambda,
	\qquad \lambda\in[-4,4].
	\]
	At $\lambda=\pm4$, the real Fermi surface  $\{p_\lambda=0\}$ collapses to an
	{elliptic critical point}.  At $\lambda=0$, it is the union of two curves
	meeting at two critical points. {This follows directly from the fact that $p_0$ is reducible,}
	\begin{equation}\label{eq:zero-fermi-factorization}
	p_0(\xi)
	=4\cos\!\left(\frac{\xi_1+\xi_2}{2}\right)
	\cos\!\left(\frac{\xi_1-\xi_2}{2}\right).
	\end{equation}
	In contrast, for
	\[
	\lambda\in(-4,0)\cup(0,4),
	\]
	the real Fermi set is one smooth curve, and its complexification is irreducible. 
	To see why this might affect things in the presence of a potential, it is most instructive to switch to Fourier space which transforms $Hu=\lambda u$,  to the following:
	\begin{equation}\label{}
	\widehat{u} = \frac{-\widehat{Vu}}{p_{\lambda}(\xi)}
	\end{equation}
	Thus, for the above to hold, the effect of $V$ on the numerator on the RHS must exactly cancel the singularity of $p_{\lambda}(\xi)$ which then translates into decay properties of $V.$ This naturally suggests that one might witness genuinely different spectral behavior depending on the singularity of  $p_{\lambda}(\xi)$ and accordingly the spectral phenomenon at  the edges, the interior critical energy ($\lambda=0$), and the regular interior energies ($\lambda\in(-4,0)\cup(0,4)$) must be treated separately.  
	The purpose of this paper is to determine the exact transition in terms of the
	asymptotic decay of $V$ and completely resolve the problem on $\Z^2$.  The
	complete phase diagram is summarized in Table~\ref{tab:sharp-d2}.
	For $n=(n_1,\ldots,n_d)\in\Z^d$, we write
	\begin{equation}\label{norms245}
	|n|=(n_1^2+\cdots+n_d^2)^{1/2},
	\qquad
	\la n\ra=(1+|n|^2)^{1/2}.
	\end{equation}
	
	\begin{table}[htbp]
		\centering
		\small
		\renewcommand{\arraystretch}{1.22}
		\begin{tabularx}{\textwidth}{@{}
				>{\centering\arraybackslash}p{0.19\textwidth}
				>{\raggedright\arraybackslash}X
				>{\raggedright\arraybackslash}p{0.31\textwidth}@{}}
			\toprule
			Energy & Sharp criterion for the potential on $\Z^2$ & The key ingredients\\
			\midrule
			$\lambda=\pm4$
			& Edge eigenfunctions exist for finitely supported potentials with any
			prescribed support size $K\geq3$, whereas support on at most two sites
			cannot create an edge eigenvalue.
			& Threshold Green function, finite differences, and moment cancellation.\\
			 
			\addlinespace
			$\lambda=0$
			& Examples exist with $|V(n)|\leq C\la n\ra^{-2}$; there is no
			eigenvalue if $|V(n)|\leq C\la n\ra^{-2-\varepsilon}$ for some
			$\varepsilon>0$.
			& An explicit construction and a Hilbert-space-valued forward iteration
			controlled by Chebyshev polynomials.  \\

			\addlinespace
			$0<|\lambda|<4$
			& Examples exist with $|V(n)|\leq C\la n\ra^{-1}$; there is no
			eigenvalue if $|V(n)|\leq C\la n\ra^{-1-\varepsilon}$ for some
			$\varepsilon>0$.
			& Wigner--von Neumann construction, discrete Carleman estimates, and forward iteration as above.  \\
		 
			\bottomrule
		\end{tabularx}
		\caption{Sharp spectral transitions for the discrete Schr\"odinger operator on
			$\Z^2$.}
		\label{tab:sharp-d2}
	\end{table}
	\FloatBarrier
	
	The second and third rows of Table~\ref{tab:sharp-d2} display the change in the
	decay scale inside the spectral band.  At the regular energies
	$0<|\lambda|<4$, the sharp scale is $\la n\ra^{-1}$, matching the scale for the continuous Schr\"odinger operator discussed shortly.  At the interior critical energy
	$\lambda=0$, the sharp lattice scale changes to $\la n\ra^{-2}$. Curiously, the same scale appears instead at the spectral edge for the continuous Schr\"odinger operator; see
	\cite{HundertmarkJexLange2023}.  Related sharp  spectral edge criteria for lattice
	Schr\"odinger operators were recently obtained by Jex and \v{S}tampach
	\cite{JexStampach2025}.
	
	The same regular--critical distinction appears, though with distinct exponents, {in \emph{limiting absorption estimates}}
	for the free lattice resolvent: weights of order
	$\la n\ra^{-1/2-\varepsilon}$ suffice near regular energies, whereas uniform
	control through critical energies requires weights of order
	$\la n\ra^{-1-\varepsilon}$ \cite{TadanoTaira2019,IsozakiKorotyaev2012}.

	Before stating the precise statements, as indicated above, we briefly compare the discrete problem with the continuous Schr\"odinger
	operator $-\Delta_{\R^d}+V$ where here $-\Delta_{\R^d}$ denotes the Laplacian.  The continuous theory of embedded eigenvalues is
	classical.  Kato's theorem gives
	absence under $V(x)=o(|x|^{-1})$ \cite{Kato1959}, while the classical
	Wigner--von Neumann example shows that  $V(x)=O(|x|^{-1})$ can support an
	embedded eigenvalue \cite{WignerVonNeumann1929}.  Both statements hold in every
	dimension.

	In contrast, sharp decay thresholds for discrete operators have been
	understood only much more recently and exhibit phenomena absent from the
	continuum, already in one dimension.  For the one-dimensional continuous
	operator, the sharp critical Coulomb coefficient (the constant that dictates the transition of the decay of the potential for the existence of an embedded eigenvalue) varies smoothly with the energy.
	For the discrete operator, however, it depends on the arithmetic type of the
	quasimomentum $k$, where $E=2\cos(\pi k)\in[-2,2]$; more precisely, it is
	continuous at irrational $k$ and discontinuous at every rational $k$
	\cite{LiuCriteria2021,ChuLiuLyu2026}.
	
	At the level of proofs, Kato's classical argument \cite{Kato1959} exploits polar
	coordinates and the decomposition of $\R^d$ into the radial variable and the
	spheres $S_r$ reducing things therefore to essentially a one dimensional problem.  The lattice $\Z^d$ is not rotationally invariant and has no
	comparable foliation by spheres.  Consequently, the discrete absence problem does
	not follow by simply discretizing the continuous proof.  More generally, many
	standard tools of Euclidean PDE have no direct lattice counterpart.

	Substantial progress on embedded eigenvalues for discrete and periodic operators
	has instead come from Floquet theory, unique continuation, complex analysis, and
	algebraic and analytic geometry.  In particular, absence results for the square lattice and for
	broader classes of perturbed periodic lattices were obtained in
	\cite{IsozakiMorioka2014,AndoIsozakiMorioka2016,AndoIsozakiMorioka2025,KuchmentVainberg2000,LiuIrreducibility2022,LiuMatosTreuer2025}.
	{These rigidity arguments rely on three principal ingredients: irreducibility of
		the Fermi variety, unique continuation, and at least exponential decay of the
		perturbation.
		These approaches use the fact that a fixed-energy Fermi surface is an algebraic
		hypersurface.  Irreducibility of the Fermi variety permits analytic cancellation or division by its
		defining polynomial, after which unique continuation rules out the resulting
		rapidly decaying solution.}
	
	The first two ingredients are essential for this mechanism.  Shipman constructed
	periodic graph operators with reducible Fermi varieties and embedded
	eigenfunctions of unbounded support \cite{Shipman2014}, while on general periodic
	graphs, compactly supported solutions may occur \cite{Kuchment1991}.  It is
	therefore natural to ask whether the third ingredient is also necessary: can
	algebraically decaying potentials suffice, and
	if so, what is the sharp decay rate?
	
	A different line of work studies quantitative unique continuation for discrete
	harmonic functions.  Buhovsky, Logunov, Malinnikova, and Sodin proved a strong
	Liouville theorem on $\Z^2$, and Bou-Rabee, Cooperman, and Ganguly subsequently
	obtained a geometric extension to periodic planar graphs
	\cite{BuhovskyLogunovMalinnikovaSodin2022,BouRabeeCoopermanGanguly2025}.
	Such discrete unique-continuation principles have important applications to
	Anderson--Bernoulli localization, as demonstrated by Ding and Smart in dimension
	two and by Li and Zhang in dimension three, following the continuum program of
	Bourgain and Kenig \cite{DingSmart2020,LiZhang2022,BourgainKenig2005}.  These
	results concern quantitative rigidity and random localization, whereas the
	present paper determines sharp decay thresholds for eigenvalues created by
	deterministic perturbations.

	\newcommand{\nin}{\noindent}	
	
	\section{Main results}
	
	We now move on to the precise statements of our main results.

	\subsection{The critical energy $0$}
	\begin{theorem}\label{thm:absence}
		Suppose that $V:\Z^2\to\mathbb C$ satisfies
		\begin{equation}\label{gaug271}
		|V(n)|\leq C\la n\ra^{-\kappa}
		\end{equation}
		for some $\kappa>2$ and $C>0$.  If
		\[
		(\Delta+V)u=0,
		\qquad
		u\in\ell^2(\Z^2),
		\]
		then $u=0$.  Consequently, $\lambda=0$ is not an eigenvalue of
		$\Delta+V$.
	\end{theorem}
	\begin{theorem}\label{thm:critical-construction}
		For every $\alpha>1$,
		there exist a real-valued potential
		$V:\Z^2\to\R$ and a nonzero function $u$ satisfying
		$|u(n)|\leq\la n\ra^{-\alpha}$ and
		\[
		(\Delta+V)u=0,
		\qquad
		|V(n)|\leq C_\alpha\la n\ra^{-2}.
		\]
		In particular, $u\in\ell^2(\Z^2)$, and hence $\lambda=0$ is an
		eigenvalue of $\Delta+V$.
	\end{theorem}
	\subsection{Spectral edge}	
	\begin{theorem}\label{thm:existence}
		Let $d\geq2$ and $K\geq2$.  For either spectral edge
		$\lambda_0\in\{-2d,2d\}$, there exist potentials $V$ supported on exactly
		$K$ sites and corresponding nonzero functions $u$ satisfying
		\begin{equation*}
		(\Delta+V)u=\lambda_0u,
		\end{equation*}
		and
		\begin{equation}\label{eq:target-decay}
		|u(n)|\leq C\la n\ra^{-(d+K-3)}.
		\end{equation}
	\end{theorem}
	
	\begin{theorem} 
		\label{thm:two-dimensional-threshold}
		Let $d=2$.
		\begin{enumerate}[label=\textup{(\roman*)}]
			\item For every integer $K\geq3$ and every
			$\lambda_0\in\{-4,4\}$, there are potentials $V$ supported on exactly
			$K$ sites such that $\lambda_0$ is an eigenvalue of $\Delta+V$.
			
			\item If $V$ is supported on at most two sites, then neither
			$4$ nor $-4$ is an eigenvalue of $\Delta+V$.
		\end{enumerate}
	\end{theorem}
	\subsection{Eigenvalues in $(-4,0)\cup(0,4)$}

	\begin{theorem}\label{thm:noncritical-construction}
		For every noncritical energy
		$\lambda\in (-4,0)\cup(0,4)$ and every $\alpha>1$, there exist a
		real-valued potential $V:\Z^2\to\R$ and a nonzero function $u$
		satisfying $|u(n)|\leq\la n\ra^{-\alpha}$ and
		\[
		(\Delta+V)u=\lambda u,
		\qquad
		|V(n)|\leq C_{\alpha,\lambda}\la n\ra^{-1}.
		\]
		In particular, $u\in\ell^2(\Z^2)$, and hence $\lambda$ is an
		eigenvalue of $\Delta+V$.
	\end{theorem}
	We now arrive at the most central result of this paper.
	\begin{theorem}\label{thm1:absence}
		Suppose that $V:\Z^2\to\mathbb C$ satisfies
		\[
		|V(n)|\leq C\la n\ra^{-\kappa}
		\]
		for some $\kappa>1$.  If
		\[
		(\Delta+V)u=\lambda u,
		\qquad \lambda\in(-4,0)\cup(0,4),
		\qquad
		u\in\ell^2(\Z^2),
		\]
		then $u=0$.  Consequently, $\lambda$ is not an eigenvalue of
		$\Delta+V$.
	\end{theorem} 
	
	A simple symmetry argument allows us to simply restrict to positive energies. Namely, the
	checkerboard transformation on $\Z^d$,
	\[
	(\mathcal Ju)(n)=(-1)^{n_1+\cdots+n_d}u(n),
	\]
	satisfies $\mathcal J\Delta\mathcal J=-\Delta$ and hence
	\begin{equation}\label{eq:checkerboard-symmetry}
	(\Delta+V)u=\lambda u
	\quad\Longleftrightarrow\quad
	(\Delta-V)\mathcal Ju=-\lambda\mathcal Ju.
	\end{equation}
	Consequently, it suffices to consider nonnegative energies. 
	\begin{remark}
 It is worth pointing out that  the case of existence of eigenvalues  outside the essential spectrum $[-4,4] $  is significantly easier,  due to the decay of the  Green's function. For instance, even carefully chosen potentials supported on a single point suffice  to create such an eigenvalue. 
	\end{remark}
	\subsection{Carleman estimate}
	
	As already indicated, a central ingredient in the proof of the above is a Carleman-type estimate which could be of independent interest and is recorded below. We need a bit of notational preparation first. 
	For $0<\lambda<4$, set
	\begin{equation}\label{criticalgrowth12}
	a_\lambda=\operatorname{arcosh}\left(1+\frac{\lambda}{2}\right).
	\end{equation}
	
	As is well known to experts, to state the estimate we first need to choose an appropriate weight function which we do next. 
	We also take this opportunity to define the setup to carry out the semiclassical analysis which lies at the heart of the estimate.

	Fix \(k\geq1\), and set \(h=k^{-1}\), and write \(x=hn\).  The scaled
	lattice is \(\Lambda_h=h\Z^2\).  Although \(x=hn\) lies in \(\Lambda_h\),
	the weights and symbols introduced below are defined on \(\mathbb R^2\),
	with \(x\) treated as a continuous variable. \ 
	Fix \(g\in\Z^2\). For $v\in \ell^2( \Lambda_h)$, define
	\begin{equation*}
	(T_gv)(x)=v(x+hg).
	\end{equation*}
	{Set \(P_\lambda:=\Delta-\lambda\). Then
		{	\begin{equation*}
			P_\lambda=\sum_{|e|_1=1}T_e-\lambda
			\end{equation*}}
		on \(\Lambda_h\) whose Fourier symbol is}
	\begin{equation*}
	p_\lambda(\xi):=2\cos\xi_1+2\cos\xi_2-\lambda,
	\qquad \xi\in\T^2.
	\end{equation*}
	
	Letting
	\begin{equation}\label{norm12}
	R_h(x)=\sqrt{h^2+|x|^2},
	\end{equation} 
	for \(\rho>0\), we define the following annuli 
	\begin{align}\label{annuli453}
	\Omega_\rho
	&=\{x\in\mathbb R^2:\rho/2\leq R_h(x)\leq4\rho\},\\
	\Omega_\rho^*
	&=\{x\in\mathbb R^2:\rho/4\leq R_h(x)\leq8\rho\}.
	\end{align}
	Since
	\begin{equation}
	R_h(hn)=h\la n\ra,
	\label{eq:scaled-radius}
	\end{equation}
	we have the following equivalence
	\begin{equation}	R_h(hn)\geq B
	\quad\Longleftrightarrow\quad
	\la n\ra\geq Bk.
	\label{eq:scaled-exterior-equivalence}
	\end{equation}	
	For \(a\geq0\), define
	\begin{equation}\label{def345}
	\phi_{a,h}(x)=aR_h(x)+\log R_h(x).
	\end{equation}
	Recalling that \(h=k^{-1}\), and since \eqref{eq:scaled-radius},
	\begin{equation}
	e^{\phi_{a,h}(hn)/h}
	=h^k e^{a\la n\ra}\la n\ra^k.
	\label{eq:weight-scaling}
	\end{equation}
	\nin	
	For \(|e|_1=1\), define
	\begin{equation}\label{gde}
	d_e(x)
	=\frac{\phi_{a,h}(x+he)-\phi_{a,h}(x)}h.
	\end{equation}
	Let \(M_g\) denote multiplication by \(g(x)\).  Then
	\begin{equation*}
	e^{\phi_{a,h}/h}T_ee^{-\phi_{a,h}/h}=M_{e^{-d_e}}T_e.
	\end{equation*}
	Define
	\begin{equation*}
	\mathcal A_{a,h}
	=e^{\phi_{a,h}/h}P_\lambda e^{-\phi_{a,h}/h}.
	\end{equation*}
	Thus
	\begin{equation}
	\mathcal A_{a,h}
	=\sum_{|e|_1=1}M_{e^{-d_e}}T_e-\lambda.
	\label{eq:exact-conjugation}
	\end{equation}
	
	Below, \(I\Subset(0,a_\lambda)\) denotes a compact interval  
	contained in \((0,a_\lambda)\).
	
	\begin{theorem}[Carleman estimates]\label{thm:carleman-estimates}
		Let \(I\Subset(0,a_\lambda)\).  There are \(B_I,C_I>0\) and \(h_I>0\)
		such that, for \(a\in I\), \(0<h\leq h_I\), \(\rho\geq B_I\), and
		every finitely supported \(v\) satisfying
		\(\operatorname{supp}v\subset\Omega_\rho\), one has
		\begin{equation}
		h\rho^{-1}\norm v_2^2
		\leq C_I\norm{\mathcal A_{a,h}v}_2^2.
		\label{eq:theorem-A-positive-a}
		\end{equation}
		There are \(B,C>0\) and \(h_0>0\) such that, for \(0<h\leq h_0\),
		\(\rho\geq B\), and every finitely supported \(v\) satisfying
		\(\operatorname{supp}v\subset\Omega_\rho\), one has
		\begin{equation}
		h\rho^{-2}\norm v_2^2
		\leq C\norm{\mathcal A_{0,h}v}_2^2.
		\label{eq:theorem-A-zero}
		\end{equation}
		
	\end{theorem}
	
	\subsection{Brief overview of the ideas}	
	We next outline the main ideas behind the proof of the three parts of the main result. \\

	\nin	
	At the spectral
	edges, the relevant object is the threshold Green function introduced in
	Subsection~\ref{subsec:edge-green-function}.  In two dimensions
	the normalized potential kernel grows logarithmically.  Taking a finite difference
	of order $K-1$ cancels the first $K-1$ terms in its far-field expansion and produces
	a solution with decay $|n|^{-(K-1)}$.  Placing the corresponding $K$ sources at
	the support of $V$ gives an edge eigenfunction for every $K\geq3$.  Conversely,
	the asymptotic {expansion of the Green function converts faster decay into moment
		conditions on those sources}.  {For two support points, the first two moment
		conditions force both source coefficients to vanish.}  This {approach is closely
		related to the threshold and ground-state questions studied} in
	\cite{HundertmarkJexLange2023,JexStampach2025}, but here the emphasis is on the
	optimal support size and the precise Green-function cancellation.
	
	At the interior critical energy $\lambda=0$, both existence and absence {reflect
		the factorization \eqref{eq:zero-fermi-factorization}.  Existence is treated in
		Section~\ref{sec:existence-zero}, and absence in
		Section~\ref{sec:absence-zero}}.  For existence, we
	construct a nonvanishing square-summable solution whose oscillation is adapted to
	one of the two components and define
	\begin{equation}\label{eq:construct-potential}
	V(n)=\frac{\lambda u(n)-\Delta u(n)}{u(n)}.
	\end{equation}
	The factorization produces the cancellation that yields
	$V(n)=O(\la n\ra^{-2})$.  For absence, we organize the equation in one
	coordinate as a Hilbert-space-valued recurrence.

	For the regular interior energies, the existence construction in
	Section~\ref{sec:existence-noncritical} belongs to the
	Wigner--von Neumann tradition, which remains an active source of sharp examples;
	see \cite{FrankSimon2017,IonescuJerison2003,Kato1959,WignerVonNeumann1929}.
	One begins with a nonvanishing oscillatory free solution along one coordinate,
	multiplies it by a carefully chosen algebraically decaying amplitude, and again
	solves \eqref{eq:construct-potential}.  
	Piecewise and gluing versions
	of this principle have produced single, multiple, and even prescribed families
	of embedded eigenvalues in one-dimensional, periodic, and geometric settings;
	see \cite{JitomirskayaLiu2019,LiuCriteria2021,LiuOng2020}.
	
	The absence theorem for $0<|\lambda|<4$ is the most delicate part of the paper.
	It begins with a Carleman estimate for the discrete operator, proved in
	Section~\ref{sec:carleman-estimates}.  Carleman estimates
	are a basic tool for unique continuation and have been used directly to exclude
	embedded eigenvalues in the continuous setting, notably by Ionescu--Jerison and
	Koch--Tataru \cite{IonescuJerison2003,KochTataru2006}.  Their adaptation to the
	lattice is not formal. 
	{Conjugating a difference operator by an exponential
		weight produces weighted translations rather than a differential operator.}
  
	We develop a semiclassical calculus for the conjugated operator
	and first obtain decay faster than any power in
	Section~\ref{sec:weighted-estimates} and then, in
	Section~\ref{sec:exponential-decay}, the crucial \emph{exponential decay}:
	\begin{equation}\label{eq:subcritical-exponential}
	e^{a\la n\ra}u\in\ell^2(\Z^2)
	\qquad\text{for every }a<a_\lambda,
	\end{equation}
	where
	\begin{equation}\label{eq:a-lambda}
	a_\lambda=\operatorname{arcosh}\!\left(1+\frac{|\lambda|}{2}\right),
	\qquad 0<|\lambda|<4.
	\end{equation}	

	The strict inequality in \eqref{eq:subcritical-exponential} is not sufficient for
	the final argument.  When the equation is iterated in one coordinate direction,
	the free recurrence grows at the exact exponential rate $e^{a_\lambda r}$.
	Thus any loss in the eigenfunction decay leaves an exponentially growing factor,
	which cannot be compensated by algebraic decay of $V$.  Reaching the endpoint
	$a=a_\lambda$ is therefore crucial.  The Carleman estimate alone
	cannot reach this endpoint, because its constants deteriorate as
	$a\uparrow a_\lambda$.  To overcome this, we use a suitable Banach space valued recursion that the restriction of the function $u$ weighted by an exponential term (with rate smaller than $e^{a_\lambda}$) to different copies of $\Z$ inside $\Z^2$ satisfies. This is solved using an ansatz involving a suitable Green kernel. We then obtain a control on the norms of the weighted functions, uniform in the exponential weight. This then allows us to take the weight to be exactly $e^{a_\lambda}$ and conclude. To prevent the exposition from getting overly technical, we refrain from presenting further elaboration in this section and more details can be found in Section \ref{sec:absence-noncritical}. \\

	\nin

	The final comment that we make before proceeding to the formal arguments is that although the complete result is stated in dimension two, several ingredients work in general dimensions including a version of the Carleman estimate. Nonetheless, what turns out to be  
	special in two
	dimensions is that the geometry of the three Fermi regimes can be classified
	completely and estimates establishing the exact transition locations can be obtained.

	\subsection{Acknowledgements} SG was supported by a Miller Research Professorship at the Miller Institute for Basic Research in Science and NSF Career grant-1945172 and NSF grant-2553225. This collaboration began while WL was a Visiting Miller Professor at the Miller Institute for Basic Research in Science at UC Berkeley in Fall 2025, and was completed during his visit to the Simons Institute for the Theory of Computing in Fall 2026. He expresses his gratitude to the Department of Mathematics at UC Berkeley, the Miller Institute, and the Simons Institute for their warm hospitality. He also thanks Yunlei Wang for bringing reference \cite{TadanoTaira2019} to his attention and June Vuong for helpful discussions.
He was supported in part by NSF grant DMS-2246031. \\

\nin
\textbf{AI usage:} An LLM was instrumental in bringing to the authors' attention the effectiveness of Carleman-type estimates in the study of spectral properties of operators and in the subsequent proof of  Theorem~\ref{thm:carleman-estimates} using semiclassical analysis. This then was the key input used to deliver the weighted estimates Propositions \ref{prop:polynomial-rigorous} and \ref{prop:exponential-rigorous}.

	\section{Absence and existence of a zero-energy eigenvalue}
	\subsection{Proof of Theorem \ref{thm:absence}}\label{sec:absence-zero}
	Write $n=(m,j)\in\Z^2$ and introduce the $\ell^2(\Z)$-valued sequence
	\[
	U_m=(u(m,j))_{j\in\Z}\in\cH:=\ell^2(\Z).
	\]
	Let $S$ be the shift on $\cH$, given by $(Sw)_j=w_{j+1}$, and set
	\[
	A=-(S+S^*).
	\]
	Recalling that $V=\{V(m,j)\}_{(m,j)\in \Z^2}$ is the potential, for each $m\in\Z$, define $V_m:\cH\to\cH$ by
	\[
	(V_mw)_j=V(m,j)w_j.
	\]
	The above is clearly well defined since $V(m,j)$ decays to $0$ as $j\to \infty.$ Given this, 
	the equation
	$(\Delta+V)u=0$ is equivalent to
	\begin{equation}\label{eq:recurrence}
	U_{m+1}+U_{m-1}-AU_m=-V_mU_m.
	\end{equation}
	The operator $A$ is self-adjoint and $\sigma(A)=[-2,2]$.
	
	Since $u\in\ell^2(\Z^2)$,
	\[
	\sum_{m\in\Z}\norm{U_m}_{\cH}^2<\infty.
	\]
	In particular, $U_m\to0$ in $\cH$ as $m\to+\infty$.  For $M\geq1$, set
	\[
	S_M:=\sup_{\ell\geq M}\norm{U_\ell}_{\cH}<\infty.
	\]
	For every $\ell\geq M$, the decay of $V$ by \eqref{gaug271} gives
	\begin{equation}\label{eq:potential-bound}
	\norm{V_\ell U_\ell}_{\cH}
	\leq C\la \ell\ra^{-\kappa}\norm{U_\ell}_{\cH}
	\leq C\la \ell\ra^{-\kappa}S_M,
	\end{equation}
	because $\la(\ell,j)\ra\geq\la \ell\ra$ for every $j\in\Z$.
	
	We now iterate \eqref{eq:recurrence} toward $+\infty$.  Define
	\begin{equation}\label{eq:kernel-recurrence}
	K_0=0,
	\qquad K_1=I,
	\qquad K_{r+1}=AK_r-K_{r-1}\quad(r\geq1).
	\end{equation}
	If $\mathcal U_k$ denotes the {Chebyshev polynomial of the second kind} \cite{WikipediaChebyshev},
	then
	\[
	K_r=\mathcal U_{r-1}(A/2).
	\]
	Since $|\mathcal U_{r-1}(x)|\leq r$ on $[-1,1]$, the spectral theorem
	implies
	\begin{equation}\label{eq:kernel-bound}
	\norm{K_r}_{\cH\to\cH}\leq r.
	\end{equation}
	Moreover, since $\kappa>2$, direct computation gives
	\begin{equation}\label{eq:weighted-tail}
	\sum_{\ell=m+1}^{\infty}(\ell-m)\la \ell\ra^{-\kappa}
	\leq C\la m\ra^{2-\kappa},
	\qquad m\geq1.
	\end{equation}
	Iterating \eqref{eq:recurrence} toward $+\infty$ formally gives
	\begin{equation}\label{eq:forward-iteration}
	U_m=-\sum_{\ell=m+1}^{\infty}K_{\ell-m}V_\ell U_\ell.
	\end{equation}
	By
	\eqref{eq:potential-bound}--\eqref{eq:weighted-tail}, the series on the
	right converges absolutely in $\cH$; for completeness the identity itself is proved in
	Subsection~\ref{subsec:partial-fourier}.
	By \eqref{eq:forward-iteration},  for every $m\geq M$,
	\begin{align}
	\norm{U_m}_{\cH}
	&\leq \sum_{\ell=m+1}^{\infty}
	\norm{K_{\ell-m}}_{\cH\to\cH}\norm{V_\ell U_\ell}_{\cH} \nonumber\\
	&\leq CS_M
	\sum_{\ell=m+1}^{\infty}(\ell-m)\la \ell\ra^{-\kappa} \nonumber\\
	&\leq C\la m\ra^{2-\kappa}S_M \nonumber\\
	&\leq C\la M\ra^{2-\kappa}S_M.\label{g18}
	\end{align}
	Taking the supremum over $m\geq M$ in \eqref{g18} yields
	\begin{equation}\label{eq:contraction}
	S_M\leq C \la M\ra^{2-\kappa}S_M.
	\end{equation}
	Because $\kappa>2$, we may choose $M$ so large that
	\[
	C\la M\ra^{2-\kappa}<1.
	\]
	It follows from \eqref{eq:contraction} that $S_M=0$.  Therefore
	\[
	U_m=0
	\qquad\text{for every }m\geq M.
	\]
	In particular,  both $U_M$ and $U_{M+1}$ vanish.
	Equation \eqref{eq:recurrence} then implies that $U_m=0$ for every
	$m\in\Z$, and consequently $u=0$.\\

	We finish with the proof of \eqref{eq:forward-iteration}.
	
	\subsection{Uniqueness}\label{subsec:partial-fourier}
	
	Let $Y_m$ denote the right-hand side of \eqref{eq:forward-iteration}.
	By
	\eqref{eq:potential-bound}--\eqref{eq:weighted-tail}, the series on the
	right converges absolutely in $\cH$ and tends to zero as $m\to\infty$.
	Direct computation shows
	\[
	Y_{m+1}+Y_{m-1}-AY_m=-V_mU_m.
	\]
	Thus $W_m:=U_m-Y_m$ satisfies 
	\begin{equation}\label{eq:free-recurrence}
	W_{m+1}+W_{m-1}-AW_m=0
	\end{equation}
	and $W_m\to0$ as $m\to\infty$.
	
	To prove \eqref{eq:forward-iteration}, it remains to show that
	$W_m=0$ for every $m\in\Z$.
	We prove that $W$ vanishes in Fourier space.
	Equip $\T$ with the normalized measure $d\theta/(2\pi)$, and use the
	same notation $\mathcal F$ for the one-dimensional unitary Fourier transform
	\[
	\mathcal F:\ell^2(\Z)\longrightarrow L^2(\T),
	\qquad
	(\mathcal Fv)(\theta)=\sum_{j\in\Z}v_j e^{-ij\theta}.
	\]

	Write
	\[
	w_m(\theta):=(\mathcal FW_m)(\theta),
	\qquad \theta\in\T.
	\]
	Under this transform, $A=-(S+S^*)$ becomes multiplication by
	$-2\cos\theta$. Therefore \eqref{eq:free-recurrence} becomes
	\begin{equation}\label{eq:scalar-free-recurrence}
	w_{m+1}(\theta)+w_{m-1}(\theta)
	+2\cos\theta\,w_m(\theta)=0.
	\end{equation}
	The two characteristic roots of \eqref{eq:scalar-free-recurrence} are
	\[
	r_+(\theta)=-e^{i\theta},
	\qquad
	r_-(\theta)=-e^{-i\theta}.
	\]
	In particular,
	\[
	|r_+(\theta)|=|r_-(\theta)|=1,
	\qquad
	r_+(\theta)-r_-(\theta)=-2i\sin\theta.
	\]
	As will be apparent shortly, it will be convenient to fix $0<\delta<1$ and consider the set
	\[
	\Omega_\delta
	:=\{\theta\in\T:|\sin\theta|\geq\delta\}.
	\]
	{For $\theta\in\Omega_\delta$,} \eqref{eq:scalar-free-recurrence} gives
	\[
	w_m(\theta)
	=d_1(\theta)r_+(\theta)^m+d_2(\theta)r_-(\theta)^m,
	\]
	where
	\[
	d_1=\frac{w_1-r_-w_0}{r_+-r_-},
	\qquad
	d_2=\frac{r_+w_0-w_1}{r_+-r_-},
	\]
	and
	\begin{equation}\label{eq:free-coefficients}
	d_1r_+^m=\frac{w_{m+1}-r_-w_m}{r_+-r_-},
	\qquad
	d_2r_-^m=\frac{r_+w_m-w_{m+1}}{r_+-r_-}.
	\end{equation}
	Since $|r_\pm|=1$ and $|r_+-r_-|\geq2\delta$ on $\Omega_\delta$,
	\eqref{eq:free-coefficients} implies
	\[
	\norm{d_\nu}_{L^2(\Omega_\delta)}
	\leq \frac{1}{2\delta}
	\bigl(\norm{w_{m+1}}_{L^2(\T)}
	+\norm{w_m}_{L^2(\T)}\bigr),
	\qquad \nu=1,2.
	\]
	By Plancherel's theorem, the right-hand side tends to zero as
	$m\to\infty$.  Hence $d_1=d_2=0$ almost everywhere on
	$\Omega_\delta$, and therefore $w_m=0$ almost everywhere on
	$\Omega_\delta$ for every $m$.  Since $\delta>0$ is arbitrary and
	$\{\theta\in\T:\sin\theta=0\}$ has measure zero, $w_m=0$ almost everywhere
	on $\T$.  The inverse Fourier transform gives $W_m=0$ for every $m$, which
	proves \eqref{eq:forward-iteration}.

	\subsection{Existence of a zero-energy eigenvalue}\label{sec:existence-zero}

	\begin{proof}[Proof of Theorem \ref{thm:critical-construction}]
		Fix any $\alpha>1$. Recalling the notation from \eqref{norms245}, define
		\[
		\eta(n)=\la n\ra^{-\alpha}.
		\]
		Set
		\[
		u(n)=(-1)^{n_2}\eta(n).
		\]
		Since $\alpha>1$,
		\[
		\sum_{n\in\Z^2}|u(n)|^2
		=\sum_{n\in\Z^2}\la n\ra^{-2\alpha}<\infty.
		\]
		Thus $u\in\ell^2(\Z^2)$.  In addition, $u(n)\neq0$ for every
		$n\in\Z^2$.
		
		Define the potential pointwise by
		\[
		V(n)=-\frac{\Delta u(n)}{u(n)}.
		\]
		It is real-valued, and its definition immediately gives
		\[
		(\Delta+V)u=0.
		\]
		
		The factor $(-1)^{n_2}$ is unchanged by shifts in the first coordinate
		and changes sign under shifts in the second coordinate.  Hence
		
		\begin{equation}\label{gaug272}
		V(n)=
		\frac{
			\eta(n+e_2)+\eta(n-e_2)
			-\eta(n+e_1)-\eta(n-e_1)
		}{\eta(n)}.
		\end{equation}

		For $\nu=1,2$,
		\[
		\frac{\eta(n\pm e_\nu)}{\eta(n)}
		=
		\left(1+\frac{\pm2n_\nu+1}{\la n\ra^2}\right)^{-\alpha/2}.
		\]
		Taylor expansion at zero, with the two signs paired, gives
		\begin{equation}\label{gaug273}
		\frac{\eta(n+e_\nu)+\eta(n-e_\nu)}{\eta(n)}
		=2+ O(\la n\ra^{-2}),
		\end{equation}
		uniformly for large $|n|$.  By \eqref{gaug272} and \eqref{gaug273}, we conclude that
		\[
		|V(n)|\leq C \la n\ra^{-2}.
		\]

	\end{proof}
	\section{Spectral edges}\label{sec:spectral-edges}

	\subsection{The Green function}\label{subsec:edge-green-function}
	
	Note that
	\[
	p(\xi):=2\sum_{\nu=1}^d(\cos\xi_\nu-1),
	\qquad \xi\in\mathbb T^d,
	\]
	is the Fourier symbol of \(\Delta-2d\).  When \(d\geq3\), we define
	\begin{equation}\label{eq:green-high-d-definition}
	G(n)=\frac{1}{(2\pi)^d}
	\int_{\mathbb T^d}\frac{e^{in\cdot\xi}}{p(\xi)}\,d\xi.
	\end{equation}
	The singularity of \(p(\xi)^{-1}\) at \(\xi=0\) is of order
	\(|\xi|^{-2}\), so this integral is finite for \(d\geq3\).  In particular,
	\begin{equation}\label{eq:green-origin-high-d}
	G(0)=\frac{1}{(2\pi)^d}
	\int_{\mathbb T^d}\frac{1}{p(\xi)}\,d\xi<0.
	\end{equation}
	When \(d=2\), the integral in
	\eqref{eq:green-high-d-definition} diverges.  We instead use the normalized
	potential kernel
	\begin{equation}\label{eq:green-two-d-definition}
	G(n)=\frac{1}{(2\pi)^2}
	\int_{\mathbb T^2}
	\frac{\cos(n\cdot\xi)-1}{p(\xi)}\,d\xi.
	\end{equation}
	{This integral is finite, and the normalization ensures that} 
	\begin{equation}\label{eq:green-origin-two-d}
	G(0)=0.
	\end{equation}
	In both cases,  
	\begin{equation}\label{eq:green-equation}
	(\Delta-2d)G=\delta_0.
	\end{equation}
	\begin{lemma}[{\cite[Theorem~1]{KozmaSchreiber2004}}]\label{lem:green-asymptotics}
		Let
		\[
		r=|n|,
		\qquad
		\omega=\frac{n}{|n|}\in \mathbb S^{d-1}.
		\]
		For every integer $N\geq1$, as $r\to\infty$,
		\begin{align}
		G(n)
		&=c_2\log r+c_0+
		\sum_{m=1}^{N}r^{-m}A_m(\omega)
		+O(r^{-N-1}),
		&& d=2,\label{eq:green-expansion-2}\\
		G(n)
		&=c_dr^{-(d-2)}
		+\sum_{m=1}^{N}r^{-(d-2+m)}A_m(\omega)
		+O\bigl(r^{-(d-2+N+1)}\bigr),
		&& d\geq3,\label{eq:green-expansion-high}
		\end{align}
		where $c_2,c_d\neq0$ and $A_1,\ldots,A_N$ are smooth functions on
		$S^{d-1}$.  The constants in the remainder estimates are uniform in
		$\omega$.
	\end{lemma}

	\subsection{Proof of Theorem~\ref{thm:existence}}
	
	\begin{proof}
		By \eqref{eq:checkerboard-symmetry}, it suffices to consider
		the upper edge $\lambda=2d$.
		Fix a large integer \(L\geq1\) and set
		\begin{equation}\label{eq:sites-and-weights}
		a_j=-jLe_1,
		\qquad
		t_j=(-1)^{K-1-j}\binom{K-1}{j},
		\qquad j=0,\ldots,K-1.
		\end{equation}
		{
			For a function \(f\) on \(\R^d\), define the forward difference
			in the \(e_1\)-direction by
			\[
			(D_L^1f)(x):=f(x+Le_1)-f(x).
			\]
			Its \((K-1)\)-fold iterate is
			\[
			D_L^{K-1}f(x)
			=\sum_{j=0}^{K-1}(-1)^{K-1-j}\binom{K-1}{j}
			f(x+jLe_1)
			=\sum_{j=0}^{K-1}t_jf(x-a_j).
			\]
			If \(f(x)=x_1^m\) with \(0\leq m\leq K-2\), then
			\(D_L^{K-1}f=0\), since each forward difference lowers the
			degree in \(x_1\) by one.  Evaluating at \(x=0\) yields
			\begin{equation}\label{eq:binomial-moments}
			\sum_{j=0}^{K-1} t_jj^m=0,
			\qquad m=0,\ldots,K-2.
			\end{equation}
		}
		Define
		\begin{equation}\label{eq:uL-definition}
		u_L(n):=\sum_{j=0}^{K-1} t_jG(n-a_j)
		=\sum_{j=0}^{K-1} t_jG(n+jLe_1).
		\end{equation}
		Using \eqref{eq:green-equation}, we immediately obtain
		\begin{equation}\label{eq:uL-source}
		(\Delta-2d)u_L=\sum_{j=0}^{K-1} t_j\delta_{a_j}.
		\end{equation}
		
		To estimate \(u_L\), take \(N=K-1\) in
		Lemma~\ref{lem:green-asymptotics}.  Every power-law term in the resulting
		expansion can be written as
		\[
		F_{\beta,A}(x)=|x|^{-\beta}
		A\left(\frac{x}{|x|}\right),
		\]
		where \(\beta>0\) and \(A\in C^\infty(\mathbb S^{d-1})\).  For every multi-index
		\(\alpha\),
		\begin{equation}\label{gF}
		|\partial^\alpha F_{\beta,A}(x)|
		\leq C_{\alpha,A}|x|^{-\beta-|\alpha|},
		\qquad |x|\geq1.
		\end{equation}
		Taylor expansion in the \(e_1\)-direction through order \(K-2\),
		together with \eqref{eq:binomial-moments} and \eqref{gF}, gives
		\begin{equation}\label{eq:finite-difference-bound}
		\left|\sum_{j=0}^{K-1} t_jF_{\beta,A}(n+jLe_1)\right|
		\leq C_L|n|^{-\beta-(K-1)}.
		\end{equation}
		Similarly, using   \eqref{eq:binomial-moments},
		\begin{equation}\label{gaug274}
		\left|\sum_{j=0}^{K-1} t_j\log|n+jLe_1|\right|
		\leq C_L|n|^{-(K-1)}.
		\end{equation}
		
		Applying \eqref{eq:finite-difference-bound} and \eqref{gaug274}
		to the finite expansion in Lemma~\ref{lem:green-asymptotics} and bounding
		the remainder at the finitely many shifted points, yields
		\begin{equation}\label{eq:uL-decay}
		|u_L(n)|\leq C_L
		\begin{cases}
		\la n\ra^{-(K-1)},&d=2,\\
		\la n\ra^{-(d+K-3)},&d\geq3.
		\end{cases}
		\end{equation}
		Thus \(u_L\) has the required decay.

		For \(j=0,\ldots,K-1\), evaluation of \eqref{eq:uL-definition} at \(a_j\)
		gives
		\begin{equation}\label{eq:uL-at-sites}
		u_L(a_j)
		=t_jG(0)
		+\sum_{\substack{0\leq\ell\leq K-1\\ \ell\neq j}}
		t_\ell G\bigl((\ell-j)Le_1\bigr).
		\end{equation}
		
		Suppose first that \(d\geq3\).  Since \(G(n)\to0\) as
		\(|n|\to\infty\),  
		\eqref{eq:uL-at-sites} implies that, as $L\to\infty$,
		\[
		u_L(a_j)=t_jG(0)+o(1).
		\]
		By \eqref{eq:green-origin-high-d}, \(G(0)<0\), and \(t_j\neq0\).
		Consequently, for large $L$, \(u_L(a_j)\neq0\) for
		$j=0,1,\ldots,K-1$.
		
		Now suppose that \(d=2\).  Our normalization gives \(G(0)=0\).
		For \(\ell\neq j\), Lemma~\ref{lem:green-asymptotics} gives
		\[
		G\bigl((\ell-j)Le_1\bigr)
		=c_2\log L+c_2\log|\ell-j|+c_0+o(1).
		\]
		Since \(\sum_{\ell=0}^{K-1}t_\ell=0\), we have
		\(\sum_{\ell\neq j}t_\ell=-t_j\).  Substitution into
		\eqref{eq:uL-at-sites} therefore gives
		\begin{equation}\label{eq:uL-at-sites-two-d}
		u_L(a_j)=-c_2t_j\log L+O(1).
		\end{equation}
		Because \(c_2t_j\neq0\), this is nonzero for sufficiently large \(L\).
		
		We conclude that for large $L$,  
		\begin{equation}\label{eq:uL-site-nonzero}
		u_L(a_j)\neq0,
		\qquad j=0,\ldots,K-1.
		\end{equation}
		
		Define
		\begin{equation}\label{eq:potential-construction}
		V(a_j)=-\frac{t_j}{u_L(a_j)},
		\qquad j=0,\ldots,K-1,
		\end{equation}
		and set \(V(n)=0\) for $n\neq a_j$.  Then
		\begin{equation}\label{eq:Vu-source}
		Vu_L=-\sum_{j=0}^{K-1} t_j\delta_{a_j}.
		\end{equation}
		Combining \eqref{eq:uL-source} and \eqref{eq:Vu-source}, we obtain
		\[
		(\Delta+V)u_L=2du_L.
		\]
		
		This completes the construction.
	\end{proof}
	
	\subsection{Proof of Theorem~\ref{thm:two-dimensional-threshold}}
	
	\begin{proof}
		By \eqref{eq:checkerboard-symmetry}, it suffices to consider
		the upper edge $\lambda=2d=4$.
		Part \textup{(i)} follows from Theorem~\ref{thm:existence}.	We now prove part \textup{(ii)}.  
		Suppose, toward a contradiction, that
		\[
		(\Delta+V)u=4u,
		\qquad
		0\neq u\in\ell^2(\Z^2),
		\]
		and \(V\) is supported on at most two sites.  Set
		\[
		f=-Vu.
		\]
		Then \(f\) is supported on at most two sites and
		\[
		(\Delta-4)u=f.
		\]
		Write $f$ as
		\[
		f=\sum_{j=1}^2 a_j\delta_{h_j},
		\]
		where $h_1$ and $h_2$ are distinct.
		Define
		\[
		w(n)=\sum_{j=1}^2 a_jG(n-h_j).
		\]
		By \eqref{eq:green-equation}, \((\Delta-4)w=f\).  Hence \(v:=u-w\)
		satisfies \((\Delta-4)v=0\).  Lemma~\ref{lem:green-asymptotics} gives
		\[
		w(n)=O(\log(2+|n|)).
		\]
		Since \(u\in\ell^2(\Z^2)\) is bounded, it follows that
		\(v(n)=O(\log(2+|n|))\).  By the standard Liouville
		theorem (discrete harmonic functions with sublinear growth must be constant), \(v\) is
		constant.  Thus, for some \(C\),
		\begin{equation}\label{gaug279}
		u(n)=C+\sum_{j=1}^2 a_jG(n-h_j).
		\end{equation}
		
		The leading Green-function asymptotic gives
		\[
		\sum_{j=1}^2 a_jG(n-h_j)
		=c_2 (a_1+a_2)\log|n|+O(1).
		\]
		Because $u\in\ell^2$, this forces 
		\begin{equation}\label{gaug278}
		a_1+a_2=0.
		\end{equation}		
		With this cancellation, Lemma~\ref{lem:green-asymptotics} also gives
		\(w(n)=O(|n|^{-1})\).   This leads to $C=0$ in \eqref{gaug279}.
		Recall that  \(\omega=n/|n|\).  Using Lemma~\ref{lem:green-asymptotics} and\(a_1+a_2=0\), we obtain
		\begin{equation}\label{eq:dipole-asymptotic}
		u(n)=-\frac{c_2}{|n|}\,
		\omega\cdot\left(\sum_{j=1}^2 a_jh_j\right)
		+O(|n|^{-2}).
		\end{equation}
		Since \(u\in\ell^2(\Z^2)\), the leading term must vanish, and hence
		\begin{equation}\label{eq:zero-first-source-moment}
		a_1h_1+a_2h_2=0.
		\end{equation}
		By \eqref{gaug278} and \eqref{eq:zero-first-source-moment}, one has that \(a_1=a_2=0\) and
		\(f=0\).
		This proves that \(\lambda=4\) cannot be an eigenvalue when
		$V$ is supported on at most two sites.
	\end{proof}
	
	\section{Existence of eigenvalues in $(-4,0)\cup(0,4)$}\label{sec:existence-noncritical}

	\begin{proof}[Proof of Theorem \ref{thm:noncritical-construction}]
		By \eqref{eq:checkerboard-symmetry}, it suffices to consider
		$\lambda\in(0,4)$.
		Write $n=(m,j)\in\Z^2$ as before.
		Choose \(k\in(0,\pi)\) such that
		\[
		\lambda=2+2\cos k,
		\]
		and choose \(\theta\) so that
		\[
		q_m:=\sin(km+\theta)\neq0\,\,\,
		{ \text{for all}}\quad(m\in\Z).
		\]
		Then, using $\sin(x)+\sin(y)=2\cos(\frac{x-y}{2})\sin(\frac{x+y}{2}),$ we have
		\begin{equation}\label{gaug2713}
		q_{m+1}+q_{m-1}+2q_m=\lambda q_m.
		\end{equation}
		Set
		\[
		w_m=q_m^2q_{m+1}^2
		\]
		and define
		\[
		\Phi(0)=0,\qquad
		\Phi(m)=
		\begin{cases}
		\displaystyle\sum_{\ell=0}^{m-1}w_\ell,&m\geq1,\\[2mm]
		\displaystyle\sum_{\ell=m}^{-1}w_\ell,&m\leq-1.
		\end{cases}
		\]
		We claim that there is a constant \(\mu>0\) such that
		\[
		\Phi(m)=\mu|m|+O(1).
		\]
		Indeed, one may check that
		\[
		\sin^2x\sin^2(x+k)
		=
		\frac14+\frac18\cos(2k)
		-\frac14\cos(2x)
		-\frac14\cos(2x+2k)
		+\frac18\cos(4x+2k).
		\]
		Substituting \(x=k\ell+\theta\) gives, when \(k\neq\pi/2\),
		\[
		w_\ell
		=
		\mu_k
		-\frac14\cos(2k\ell+2\theta)
		-\frac14\cos(2k\ell+2\theta+2k)
		+\frac18\cos(4k\ell+4\theta+2k),
		\]
		where
		\[
		\mu_k=\frac14+\frac18\cos(2k)\geq\frac18.
		\]
		The partial sums of all three nonconstant cosine terms are uniformly
		bounded.  Hence
		\[
		\sum_{\ell=0}^{M-1}w_\ell=\mu_kM+O(1).
		\]
		If \(k=\pi/2\), then
		\[
		w_\ell=\sin^2\theta\cos^2\theta>0
		\]
		for every \(\ell\), since \(q_\ell\neq0\) for all \(\ell\).  
		
		Fix \(\alpha>1\) and define
		\[
		F(m,j)=1+\Phi(m)+|j|,
		\qquad
		\eta(m,j)=F(m,j)^{-\alpha},
		\qquad
		u(m,j)=\eta(m,j)q_m.
		\]
		Since \(q_m\neq0\), we may define
		\[
		V(m,j)=\lambda-\frac{(\Delta u)(m,j)}{u(m,j)}.
		\]
		Equivalently,
		\begin{align}
		V(m,j)=\lambda
		&-\frac{\eta(m+1,j)q_{m+1}}{\eta(m,j)q_m}
		-\frac{\eta(m-1,j)q_{m-1}}{\eta(m,j)q_m}\\
		&-\frac{\eta(m,j+1)}{\eta(m,j)}
		-\frac{\eta(m,j-1)}{\eta(m,j)}.\label{gaug2714}
		\end{align}
		By the construction,
		\[
		(\Delta+V)u=\lambda u.
		\]
		
		We now estimate \(V\).  
		A direct computation shows
		\begin{equation}\label{gaug2710}
		\left|
		\left(\frac{\eta(m+1,j)}{\eta(m,j)}-1\right)
		\frac{q_{m+1}}{q_m}
		\right|
		\leq
		\frac{C}{F(m,j)}
		|q_m|\,|q_{m+1}|^3
		\leq\frac{C}{F(m,j)},
		\end{equation}
		\begin{equation}\label{gaug2711}
		\left|
		\left(\frac{\eta(m-1,j)}{\eta(m,j)}-1\right)
		\frac{q_{m-1}}{q_m}
		\right|
		\leq
		\frac{C}{F(m,j)}
		|q_m|\,|q_{m-1}|^3
		\leq\frac{C}{F(m,j)},
		\end{equation}
		and 
		\begin{equation}\label{gaug2712}
		\left|
		\frac{\eta(m,j\pm 1)}{\eta(m,j)}
		-\frac{\eta(m,j)}{\eta(m,j)}
		\right|
		\leq\frac{C}{F(m,j)}.
		\end{equation}
		
		By \eqref{gaug2713}, \eqref{gaug2714} and \eqref{gaug2710}-\eqref{gaug2712},  we conclude that
		\[
		|V(m,j)|
		\leq\frac{C}{F(m,j)}
		\leq C\la(m,j)\ra^{-1}.
		\]
		
		Moreover,
		\[
		|u(m,j)|
		\leq C\la(m,j)\ra^{-\alpha}.
		\]
		Multiplying $u$ by a nonzero constant does not change $V$, so we may
		assume that $|u(n)|\leq\la n\ra^{-\alpha}$.
		Since \(\alpha>1\), this implies
		\[
		u\in\ell^2(\Z^2).
		\]

	\end{proof}

	\section{Carleman estimates}\label{sec:carleman-estimates}
	The rest of the paper is devoted to the proof of Theorem \ref{thm1:absence}. As already indicated, the first key step is to deliver the Carleman estimate recorded in Theorem \ref{thm:carleman-estimates}.
	
	{
		The proof of Theorem~\ref{thm:carleman-estimates} follows the classical
		positive commutator approach to Carleman estimates, using semiclassical
		analysis.  We first write the conjugated operator as
		\(\mathcal A_{a,h}=\operatorname{Op}_h(A_{a,h})\), where \(A_{a,h}\)
		is its symbol.  We then introduce the leading comparison symbol
		\(A_{a,h}^0(x,\xi)=p_\lambda(\xi+i\nabla\phi_{a,h}(x))\),
		whose explicit expression is simpler.  Estimating the error reduces
		our problem to studying the comparison operator
		\(\mathcal A_{a,h}^0=\operatorname{Op}_h(A_{a,h}^0)\).
		
		Next, we split \(\mathcal A_{a,h}^0\) into self-adjoint real and
		imaginary parts.  Their commutator, multiplied by \(i/h\), is close
		to the quantization of the Poisson bracket of the real and imaginary
		parts of \(A_{a,h}^0\).  We calculate this bracket and show that, after a suitable modification and rescaling, the resulting expression is uniformly positive {which after some routine rearranging yields the statement of Theorem \ref{thm:carleman-estimates}.}}
	
	\subsection{Quantization}
	
	{ We rely on the analogue of the standard semiclassical quantization for the rescaled
		lattice \(\Lambda_h=h\Z^2\);
		see \cite[Chapter~4]{Zworski2012} for the corresponding continuous
		construction.}

	Let \(\Gamma\subset\Z^2\) be finite, {and let the symbol
		\(b=b(x,\xi): \R^2 \times \T^2 \to \C \)}  have the finite Fourier expansion
	\begin{equation*}
	b(x,\xi)=\sum_{g\in\Gamma}\widehat b(x,g)e^{ig\cdot\xi},
	\qquad
	\widehat b(x,g)
	=\frac1{(2\pi)^2}\int_{\T^2}b(x,\xi)e^{-ig\cdot\xi}\,d\xi.
	\end{equation*}
	{
		For a finitely supported function \(v\) on \(\Lambda_h\), define the
		quantization of \(b\) by
		\begin{equation*}
		\bigl(\operatorname{Op}_h(b)v\bigr)(x)
		=\frac1{(2\pi)^2}\int_{\T^2}
		\sum_{y\in\Lambda_h}e^{i(x-y)\cdot\xi/h}
		b(x,\xi)v(y)\,d\xi,
		\qquad x\in\Lambda_h.
		\end{equation*}
		Writing \(x=hn\) and \(y=hm\), we have
		\begin{equation*}
		\frac1{(2\pi)^2}\int_{\T^2}
		e^{i(n-m+g)\cdot\xi}\,d\xi
		=\delta_{m,n+g}.
		\end{equation*}}
	Consequently, the Fourier expansion of \(b\) gives
	\begin{equation}\label{gsep11}
	\bigl(\operatorname{Op}_h(b)v\bigr)(x)
	=\sum_{g\in\Gamma}\widehat b(x,g)v(x+hg).
	\end{equation}
	Equivalently,
	\begin{equation*}
	\operatorname{Op}_h(b)
	=\sum_{g\in\Gamma}M_{\widehat b(x,g)}T_g.
	\end{equation*}

	Notice that the quantization of a real-valued symbol need not
	be self-adjoint.  Indeed, if \(b\) is real-valued, then
	\begin{equation*}
	\operatorname{Op}_h(b)^*
	=\sum_{g\in\Gamma}M_{\widehat b(x+hg,g)}T_g,
	\end{equation*}
	{which generally differs from} \(\operatorname{Op}_h(b)\).   
	
	For real-valued symbols \(f,g\), set 	\begin{equation*}
	\{f,g\}
	=\partial_\xi f\cdot\partial_xg
	-\partial_xf\cdot\partial_\xi g,
	\end{equation*}
	where
	\[
	\partial_x=(\partial_{x_1},\partial_{x_2}),
	\qquad
	\partial_\xi=(\partial_{\xi_1},\partial_{\xi_2}).
	\]

	In view of \eqref{eq:exact-conjugation}, define
	\begin{equation}\label{g11}
	A_{a,h}(x,\xi)
	=\sum_{|e|_1=1}e^{-d_e(x)}e^{ie\cdot\xi}-\lambda,
	\end{equation}
	so that
	\begin{equation*}
	\mathcal A_{a,h}=\operatorname{Op}_h(A_{a,h}).
	\end{equation*}
	Define the comparison symbol directly by
	\begin{equation}
	A_{a,h}^0(x,\xi)
	=\sum_{|e|_1=1}e^{-e\cdot\nabla\phi_{a,h}(x)}
	e^{ie\cdot\xi}-\lambda.
	\label{eq:comparison-symbol}
	\end{equation}
	Set
	\begin{equation*}
	\mathcal A_{a,h}^0
	=\operatorname{Op}_h(A_{a,h}^0)
	=\sum_{|e|_1=1}M_{e^{-e\cdot\nabla\phi_{a,h}}}T_e-\lambda.
	\end{equation*}
	The next lemma compares $\mathcal A_{a,h}$ and $\mathcal A_{a,h}^0.$\\

	The following derivative estimates will be used.  For every multi-index \(\beta\) with \(|\beta|\geq1\),
	\begin{equation}
	|\partial_x^\beta R_h(x)|
	\leq C_\beta R_h(x)^{1-|\beta|},
	\qquad
	|\partial_x^\beta\log R_h(x)|
	\leq C_\beta R_h(x)^{-|\beta|}.
	\label{eq:radius-derivative-bounds}
	\end{equation}
	Recalling \eqref{def345}, by \eqref{eq:radius-derivative-bounds}  and \eqref{gde}, for every \(a\geq0\),
	\begin{equation}
	|\partial_x^\beta\phi_{a,h}(x)|
	\leq C_\beta\left(
	aR_h(x)^{1-|\beta|}+R_h(x)^{-|\beta|}\right),
	\qquad |\beta|\geq1.
	\label{eq:weight-derivative-bounds}
	\end{equation}
	
	\begin{lemma}
		\label{lem:comparison-symbols}
		Fix \(\bar a>0\).  There are \(B_{\bar a},C_{\bar a}>0\) such that, for
		\(0\leq a\leq\bar a\), \(0<h\leq1\), \(\rho\geq B_{\bar a}\),
		and \(v\in\ell^2(\Lambda_h)\) with
		\(\operatorname{supp}v\subset\Omega_\rho^*\), one has
		\begin{equation}
		\norm{(\mathcal A_{a,h}-\mathcal A_{a,h}^0)v}_2
		\leq C_{\bar a}h\rho^{-1}\norm v_2.
		\label{eq:comparison-operator-error}
		\end{equation}
		There are also \(B,C>0\), independent of \(\bar a\), such that, for
		\(0<h\leq1\), \(\rho\geq B\), and \(v\in\ell^2(\Lambda_h)\) with
		\(\operatorname{supp}v\subset\Omega_\rho^*\),
		\begin{equation}
		\norm{(\mathcal A_{0,h}-\mathcal A_{0,h}^0)v}_2
		\leq Ch\rho^{-2}\norm v_2.
		\label{eq:comparison-operator-error-zero}
		\end{equation}
	\end{lemma}
	
	\begin{proof}
		For \(|e|_1=1\), put
		\begin{equation*}
		d_e^0(x)=e\cdot\nabla\phi_{a,h}(x).
		\end{equation*}
		{By the fundamental theorem of calculus,
			\begin{equation}
			d_e(x)=\int_0^1e\cdot\nabla\phi_{a,h}(x+the)\,dt.
			\label{eq:discrete-slope-integral}
			\end{equation}}		Taylor's formula at \(x\) gives
		\begin{equation*}
		d_e(x)-d_e^0(x)
		=h\int_0^1(1-t)e^TD^2\phi_{a,h}(x+the)e\,dt.
		\end{equation*}
		If \(v(x+he)\neq0\), then \(x+he\in\Omega_\rho^*\), and hence
		\begin{equation*}
		\rho/4-h\leq R_h(x+the)\leq8\rho+h,
		\qquad 0\leq t\leq1.
		\end{equation*}
		Thus, by \eqref{eq:weight-derivative-bounds},
		\begin{equation}
		|d_e(x)-d_e^0(x)|
		\leq C_{\bar a}h\rho^{-1}.
		\label{eq:discrete-slope-remainder-bound}
		\end{equation}
		For \(a=0\), the same argument gives
		\begin{equation}
		|d_e(x)-d_e^0(x)|
		\leq Ch\rho^{-2}.
		\label{eq:discrete-slope-remainder-bound-zero}
		\end{equation}
		Moreover, \eqref{eq:discrete-slope-integral} and
		\eqref{eq:weight-derivative-bounds} imply
		\begin{equation}\label{g12}
		|d_e(x)|+|d_e^0(x)|\leq C_{\bar a}.
		\end{equation}
		When \(a=0\), the same estimate holds with a constant \(C\)
		independent of \(\bar a\).
		By \eqref{g11} and \eqref{eq:comparison-symbol},
		\begin{equation}\label{g13}
		\bigl((\mathcal A_{a,h}-\mathcal A_{a,h}^0)v\bigr)(x)
		=\sum_{|e|_1=1}
		\left(e^{-d_e(x)}-e^{-d_e^0(x)}\right)v(x+he).
		\end{equation}
		By \eqref{eq:discrete-slope-remainder-bound},
		\eqref{eq:discrete-slope-remainder-bound-zero}, \eqref{g12}, and
		\eqref{g13}, we obtain \eqref{eq:comparison-operator-error} and
		\eqref{eq:comparison-operator-error-zero}.
	\end{proof}
	For future applications, we next record necessary estimates for $A_{a,h}^0,$ the symbol corresponding to ${\mathcal A_{a,h}^0}.$
	
	\begin{proposition} 
		\label{prop:comparison-coefficient-estimates}
		
		Set
		\begin{equation}\label{steps}
		\Gamma_1=\{0,\pm e_1,\pm e_2\}.
		\end{equation}
		Write
		{	\begin{equation}\label{notation34765}
			A_{a,h}^0(x,\xi)
			=\sum_{e\in\Gamma_1}a_e(x)e^{ie\cdot\xi},
			\qquad
			a_0=-\lambda,\qquad
			a_e=e^{-e\cdot\nabla\phi_{a,h}}
			\quad\text{if }|e|_1=1.
			\end{equation}}
		Fix \(\bar a>0\).  There exists \(B_{\bar a}\geq1\) such that, for every
		multi-index \(\beta\), there is a constant \(C_{\bar a,\beta}>0\) for which, whenever
		\(0\leq a\leq\bar a\), \(0<h\leq1\), \(\rho\geq B_{\bar a}\), and
		\(x\in\Omega_\rho^*\), one has
		\begin{equation}
		|\partial_x^\beta a_e(x)|
		\leq C_{\bar a,\beta}\rho^{-|\beta|},
		\qquad e\in\Gamma_1.
		\label{eq:comparison-coefficient-estimates}
		\end{equation}
		
		Let
		\begin{equation*}
		q_{a,h}=\operatorname{Re}A_{a,h}^0,
		\qquad
		s_{a,h}=\operatorname{Im}A_{a,h}^0.
		\end{equation*}
		There is also \(B\geq1\), independent of \(\bar a\), such that, for every
		multi-index \(\beta\), there is \(C_\beta>0\) for which the following
		estimates hold when \(a=0\), \(0<h\leq1\), \(\rho\geq B\), and 
		\(x\in\Omega_\rho^*\):
		\begin{align*}
		q_{0,h}(x,\xi)-p_\lambda(\xi)
		&=2\sum_{j=1}^2
		\bigl(\cosh(\partial_{x_j}\phi_{0,h}(x))-1\bigr)\cos\xi_j,\\
		s_{0,h}(x,\xi)
		&=-2\sum_{j=1}^2
		\sinh(\partial_{x_j}\phi_{0,h}(x))\sin\xi_j,
		\end{align*}
		and hence
		{\begin{align}
			\left|\partial_x^\beta\left(
			q_{0,h}(x,\xi)-p_\lambda(\xi)\right)\right|
			&\leq C_\beta\rho^{-2-|\beta|},
			\notag\\
			|\partial_x^\beta s_{0,h}(x,\xi)|
			&\leq C_\beta\rho^{-1-|\beta|}.
			\label{eq:zero-comparison-symbol-derivatives}
			\end{align}}
	\end{proposition}
	
	\begin{proof}
		By \eqref{eq:radius-derivative-bounds} and
		\eqref{eq:weight-derivative-bounds},
		\begin{equation*}
		|e\cdot\nabla\phi_{a,h}(x)|\leq C_{\bar a},
		\qquad
		|\partial_x^\beta(e\cdot\nabla\phi_{a,h}(x))|
		\leq C_{\bar a,\beta}\rho^{-|\beta|},
		\qquad |\beta|\geq1.
		\end{equation*}
		The chain rule applied to
		\(a_e=e^{-e\cdot\nabla\phi_{a,h}}\) proves
		\eqref{eq:comparison-coefficient-estimates}.
		
		If \(a=0\), \eqref{eq:weight-derivative-bounds}, applied with
		derivative order \(|\beta|+1\), gives
		\begin{equation}\label{g40}
		|\partial_x^\beta(e\cdot\nabla\phi_{0,h}(x))|
		\leq C_\beta\rho^{-1-|\beta|}.
		\end{equation}
		Pairing the terms indexed by \(e_j\) and \(-e_j\) gives the
		displayed formulas for \(q_{0,h}\) and \(s_{0,h}\).
		The estimate \eqref{eq:zero-comparison-symbol-derivatives} follows from \eqref{g40}. 
	\end{proof}
	
	\subsection{Symmetrization}
	{ Even when a symbol \(b\) is nonnegative, its quantization
		\(\operatorname{Op}_h(b)\) need not be a nonnegative operator.
		The following lemma compares this operator with a nonnegative quadratic
		form and estimates the difference, so that positivity of the symbol
		yields  a  lower bound of  the corresponding quantization.}
	
	\begin{lemma} 
		\label{lem:friedrichs}
		Let \(\Gamma\subset\Z^2\) be fixed and finite, let \(\alpha\geq0\), and
		let
		\begin{equation}\label{symbol5678}
		b(x,\xi)=\sum_{g\in\Gamma}\widehat b(x,g)e^{ig\cdot\xi}.
		\end{equation}
		There exists \(B_\Gamma\geq1\) such that the following holds.  Let
		\(\rho\geq B_\Gamma\), and suppose
		\begin{equation*}
		b(x,\xi)\geq0,
		\qquad (x,\xi)\in\Omega_\rho^*\times\T^2.
		\end{equation*}
		Suppose also that, for every \(g\in\Gamma\) and every multi-index
		\(\beta\) with \(|\beta|\leq1\),
		\begin{equation}
		|\partial_x^\beta\widehat b(x,g)|
		\leq C\rho^{-\alpha-|\beta|}
		\label{eq:friedrichs-coefficient-bounds}
		\end{equation}
		on \(\Omega_\rho^*\).  Then, for \(0<h\leq1\) and
		\(\operatorname{supp}v\subset\Omega_\rho\),
		\begin{equation}
		\operatorname{Re}\ip{\operatorname{Op}_h(b)v}{v}
		\geq-C_\Gamma\left(
		\rho^{-\alpha-1}+h^2\rho^{-\alpha}\right)\norm v_2^2.
		\label{eq:friedrichs-conclusion}
		\end{equation}
		Here \(C_\Gamma\) is independent of \(h,\rho\); it may depend on
		\(\Gamma,\alpha\), and the constant in
		\eqref{eq:friedrichs-coefficient-bounds}.
	\end{lemma}
	\begin{proof}
		One may note that in \eqref{symbol5678} if $\widehat b(x,g)$ does not depend on $x$ then clearly $\operatorname{Op}_h(b)$ is self-adjoint and is nonnegative under the assumption of nonnegativity of $b$. The argument now proceeds by approximating \(\operatorname{Op}_h(b)\) by an average of nonnegative forms.

		Choose a nonnegative function
		\(\chi\in C_c^\infty(\mathbb R^2)\), supported in the unit ball,
		such that \(\int\chi^2=1\), and set
		\begin{equation*}
		\chi_y(x)=\chi(x-y).
		\end{equation*}
		Then
		\begin{equation}\label{g30}
		\int_{\mathbb R^2}\chi_y(x)^2\,dy=1.
		\end{equation}
		By \eqref{eq:radius-derivative-bounds} and
		\(\operatorname{supp}v\subset\Omega_\rho\), the condition
		\(\chi_yv\neq0\) implies that \(y\in\Omega_\rho^*\).
		For each \(y\in\Omega_\rho^*\), let
		\begin{equation*}
		B_y=\sum_{g\in\Gamma}\widehat b(y,g)T_g.
		\end{equation*}
		Since \(b(y,\xi)\geq0\), the lattice Fourier transform gives
		\begin{equation*}
		\ip{B_yu}{u}
		=\frac1{(2\pi)^2}\int_{\T^2}
		b(y,\xi)|\widehat u(\xi)|^2\,d\xi\geq0.
		\end{equation*}
		Therefore the quadratic form
		\begin{equation*}
		\ip{B_h^Fv}{v}
		:=\int_{\mathbb R^2}
		\ip{B_y(\chi_yv)}{\chi_yv}\,dy
		\end{equation*}
		is nonnegative.	Expanding this form gives
		\begin{equation*}
		\ip{B_h^Fv}{v}
		=\sum_{x,g}\widehat b^{\,F}(x,g)
		v(x+hg)\overline{v(x)},
		\end{equation*}
		where 		\begin{equation*}
		\widehat b^{\,F}(x,g)
		=\int_{\mathbb R^2}\widehat b(y,g)
		\chi_y(x)\chi_y(x+hg)\,dy.
		\end{equation*}
		Put
		\begin{equation*}
		\kappa_g=\int_{\mathbb R^2}
		\chi_y(x)\chi_y(x+hg)\,dy.
		\end{equation*}
		By \eqref{g30},
		\begin{equation*}
		|1-\kappa_g|
		=\frac12\int_{\mathbb R^2}
		|\chi_y(x+hg)-\chi_y(x)|^2\,dy
		\leq C_\Gamma h^2.
		\end{equation*}
		Also, on the support of
		\(\chi_y(x)\chi_y(x+hg)\), we have \(|x-y|\leq1\).  Therefore
		\eqref{eq:friedrichs-coefficient-bounds} implies
		\begin{equation}\label{gsep12}
		\begin{aligned}
		|\widehat b^{\,F}(x,g)-\widehat b(x,g)|
		&\leq\int_{\mathbb R^2}
		|\widehat b(y,g)-\widehat b(x,g)|
		\chi_y(x)\chi_y(x+hg)\,dy\\
		&\quad+|1-\kappa_g|\,|\widehat b(x,g)|\\
		&\leq C_\Gamma\left(
		\rho^{-\alpha-1}+h^2\rho^{-\alpha}\right).
		\end{aligned}
		\end{equation}
		{Since \(\Gamma\) is finite}, \eqref{gsep11} and \eqref{gsep12} imply
		\begin{equation*}
		\left|
		\ip{(B_h^F-\operatorname{Op}_h(b))v}{v}
		\right|
		\leq C_\Gamma\left(
		\rho^{-\alpha-1}+h^2\rho^{-\alpha}\right)\norm v_2^2.
		\end{equation*}
		The nonnegativity of \(B_h^F\) now proves
		\eqref{eq:friedrichs-conclusion}.
	\end{proof}

	\subsection{Operator comparisons}
	Set
	\begin{equation}\label{abbrev67}
	A=A_{a,h}^0=q+is,
	\qquad
	\mathcal A=\mathcal A_{a,h}^0,
	\qquad
	\mathcal Q=\frac{\mathcal A+\mathcal A^*}{2},
	\qquad
	\mathcal S=\frac{\mathcal A-\mathcal A^*}{2i}.
	\end{equation}
	Then both \(\mathcal Q\) and \(\mathcal S\) are self-adjoint, and
	\begin{align}
	\ip{i[\mathcal Q,\mathcal S]v}{v}
	&=\frac12\left(
	\norm{\mathcal Av}_2^2-\norm{\mathcal A^*v}_2^2\right),
	\notag\\
	\norm{\mathcal Qv}_2^2+\norm{\mathcal Sv}_2^2
	&=\frac12\left(
	\norm{\mathcal Av}_2^2+\norm{\mathcal A^*v}_2^2\right),
	\label{eq:quadratic-exact-identity}
	\end{align}
	
	{	The next lemma contains the comparison estimates between self-adjoint versions of the operator $\mathcal A$ and quantization of natural counterparts of its symbol. }

	\begin{lemma}[Operator comparison]
		\label{lem:operator-comparison}
		Fix \(\bar a>0\).  There are \(B_{\bar a},C_{\bar a}>0\) such that,
		whenever \(0\leq a\leq\bar a\), \(0<h\leq1\),
		\(\rho\geq B_{\bar a}\), and
		\(\operatorname{supp}v\subset\Omega_\rho\), one has
		\begin{align}
		\norm{(\mathcal Q-\operatorname{Op}_h(q))v}_2
		+\norm{(\mathcal S-\operatorname{Op}_h(s))v}_2
		&\leq C_{\bar a}h\rho^{-1}\norm v_2,
		\notag\\
		\norm{(\mathcal Q^2-\operatorname{Op}_h(q^2))v}_2
		+\norm{(\mathcal S^2-\operatorname{Op}_h(s^2))v}_2
		&\leq C_{\bar a}h\rho^{-1}\norm v_2.
		\label{eq:quadratic-operator-comparison}
		\end{align}
		Moreover,
		\begin{equation}
		\norm{\left(
			i[\mathcal Q,\mathcal S]-h\operatorname{Op}_h(\{q,s\})
			\right)v}_2
		\leq C_{\bar a}h^2\rho^{-2}\norm v_2.
		\label{eq:commutator-operator-comparison}
		\end{equation}
		There are also \(B,C>0\), independent of \(\bar a\), such that, when
		\(a=0\), \(0<h\leq1\), \(\rho\geq B\), and
		\(\operatorname{supp}v\subset\Omega_\rho\), the estimates in
		\eqref{eq:quadratic-operator-comparison} improve to
		{\begin{align}
			\norm{(\mathcal Q-\operatorname{Op}_h(q))v}_2
			+\norm{(\mathcal S-\operatorname{Op}_h(s))v}_2
			&\leq Ch\rho^{-2}\norm v_2,
			\notag\\
			\rho^{-2}\norm{(\mathcal Q^2-\operatorname{Op}_h(q^2))v}_2
			+\norm{(\mathcal S^2-\operatorname{Op}_h(s^2))v}_2
			&\leq Ch\rho^{-3}\norm v_2.
			\label{eq:quadratic-operator-comparison-zero}
			\end{align}}
		In this case, \eqref{eq:commutator-operator-comparison} also holds with
		\(C\) in place of \(C_{\bar a}\).
	\end{lemma}
	
	\begin{proof}
		Put
		\begin{equation*}
		\mathcal Q_1=\operatorname{Op}_h(q),
		\qquad
		\mathcal S_1=\operatorname{Op}_h(s).
		\end{equation*}
		Thus \(\mathcal A=\mathcal Q_1+i\mathcal S_1\), whereas
		\(\mathcal A=\mathcal Q+i\mathcal S\) with \(\mathcal Q,\mathcal S\)
		self-adjoint.  Therefore
		\begin{align}
		\mathcal Q-\mathcal Q_1
		&=\frac12(\mathcal Q_1^*-\mathcal Q_1)
		+\frac i2(\mathcal S_1-\mathcal S_1^*),
		\notag\\
		\mathcal S-\mathcal S_1
		&=\frac1{2i}(\mathcal Q_1-\mathcal Q_1^*)
		+\frac12(\mathcal S_1^*-\mathcal S_1).
		\label{eq:QS-versus-left-quantizations}
		\end{align}
		
		Recalling $\Gamma_1$ from \eqref{steps}, write
		\begin{equation}\label{symbol12}
		q(x,\xi)=\sum_{e\in\Gamma_1}q_e(x)e^{ie\cdot\xi},
		\qquad
		s(x,\xi)=\sum_{e\in\Gamma_1}s_e(x)e^{ie\cdot\xi}.
		\end{equation}
		We use the following direct calculation.  If
		\(b(x,\xi)=\sum_eb_e(x)e^{ie\cdot\xi}\) is real-valued, then
		\(b_{-e}=\overline{b_e}\), and hence
		\begin{equation}
		\bigl((\operatorname{Op}_h(b)^*
		-\operatorname{Op}_h(b))v\bigr)(x)
		=\sum_e\bigl(b_e(x+he)-b_e(x)\bigr)v(x+he).
		\label{eq:real-symbol-adjoint-difference}
		\end{equation}
		By Proposition~\ref{prop:comparison-coefficient-estimates},
		\begin{equation*}
		|q_e(x+he)-q_e(x)|+|s_e(x+he)-s_e(x)|
		\leq C_{\bar a}h\rho^{-1}.
		\end{equation*}
		Thus
		\eqref{eq:real-symbol-adjoint-difference} yields
		\begin{equation*}
		\norm{(\mathcal Q_1^*-\mathcal Q_1)v}_2
		+\norm{(\mathcal S_1^*-\mathcal S_1)v}_2
		\leq C_{\bar a}h\rho^{-1}\norm v_2.
		\end{equation*}
		Together with \eqref{eq:QS-versus-left-quantizations}, this proves
		the first estimate in \eqref{eq:quadratic-operator-comparison}.
		
		For the square terms, direct multiplication gives
		\begin{equation}
		\bigl((\operatorname{Op}_h(b)^2
		-\operatorname{Op}_h(b^2))v\bigr)(x)
		=\sum_{e,f}b_e(x)
		\bigl(b_f(x+he)-b_f(x)\bigr)v(x+h(e+f)).
		\label{eq:left-quantization-square-difference}
		\end{equation}
		It follows from
		Proposition~\ref{prop:comparison-coefficient-estimates} that
		\begin{equation*}
		\norm{(\mathcal Q_1^2-\operatorname{Op}_h(q^2))v}_2
		+\norm{(\mathcal S_1^2-\operatorname{Op}_h(s^2))v}_2
		\leq C_{\bar a}h\rho^{-1}\norm v_2.
		\end{equation*}
		Moreover, \(\mathcal Q_1,\mathcal S_1,\mathcal Q,\mathcal S\) are
		uniformly bounded and
		\begin{equation*}
		\mathcal Q^2-\mathcal Q_1^2
		=\mathcal Q(\mathcal Q-\mathcal Q_1)
		+(\mathcal Q-\mathcal Q_1)\mathcal Q_1,
		\qquad
		\mathcal S^2-\mathcal S_1^2
		=\mathcal S(\mathcal S-\mathcal S_1)
		+(\mathcal S-\mathcal S_1)\mathcal S_1.
		\end{equation*}
		The first estimate in \eqref{eq:quadratic-operator-comparison}
		therefore gives
		\begin{equation*}
		\norm{(\mathcal Q^2-\mathcal Q_1^2)v}_2
		+\norm{(\mathcal S^2-\mathcal S_1^2)v}_2
		\leq C_{\bar a}h\rho^{-1}\norm v_2.
		\end{equation*}
		This proves the second estimate in
		\eqref{eq:quadratic-operator-comparison}.\\

		\nin
		We now prove \eqref{eq:commutator-operator-comparison}.  
		Recall from \eqref{abbrev67} and \eqref{notation34765} that
		\begin{equation*}
		A(x,\xi)=\sum_{e\in\Gamma_1}a_e(x)e^{ie\cdot\xi}.
		\end{equation*}
		Since only finitely many shifts occur, all coefficient values used
		below lie in \(\Omega_\rho^*\).  The identities
		\begin{equation*}
		(M_\alpha T_e)(M_\beta T_f)
		=M_{\alpha(x)\beta(x+he)}T_{e+f},
		\qquad
		(M_\alpha T_e)^*
		=M_{\overline{\alpha(x-he)}}T_{-e}
		\end{equation*}
		and Taylor's formula give
		\begin{align}
		\mathcal A^*\mathcal A
		&=\operatorname{Op}_h\left(
		|A|^2+\frac h i\left[
		A\sum_{j=1}^2
		\partial_{x_j}\partial_{\xi_j}\overline A
		+\partial_\xi\overline A\cdot\partial_xA
		\right]\right)+R_1,\label{gaug2721}\\
		\mathcal A\mathcal A^*
		&=\operatorname{Op}_h\left(
		|A|^2+\frac h i\left[
		A\sum_{j=1}^2
		\partial_{x_j}\partial_{\xi_j}\overline A
		+\partial_\xi A\cdot\partial_x\overline A
		\right]\right)+R_2,\label{gaug2722}
		\end{align}
		where Proposition~\ref{prop:comparison-coefficient-estimates}
		implies
		\begin{equation*}
		\norm{R_1v}_2+\norm{R_2v}_2
		\leq C_{\bar a}h^2\rho^{-2}\norm v_2.
		\end{equation*}
		When \(a=0\), the constant on the right-hand side can be chosen
		independently of \(\bar a\).
		Since
		\begin{equation*}
		i[\mathcal Q,\mathcal S]=\frac12\left(
		\mathcal A^*\mathcal A-\mathcal A\mathcal A^*\right)
		\end{equation*}
		and \(\{\overline A,A\}=2i\{q,s\}\), subtracting the two
		expansions  (\eqref{gaug2721} and \eqref{gaug2722}) gives
		\begin{equation*}
		i[\mathcal Q,\mathcal S]
		=h\operatorname{Op}_h(\{q,s\})
		+\frac12(R_1-R_2).
		\end{equation*}
		This proves \eqref{eq:commutator-operator-comparison}.

		For \(a=0\), the proof is similar, using the sharper estimates
		\eqref{eq:zero-comparison-symbol-derivatives}.
	\end{proof}
	\subsection{Positivity of symbols} 
	We now prove the following proposition.  
	\begin{proposition}\label{thm:symbol-positivity}
		Let \(I\Subset(0,a_\lambda)\).  There are \(B_I,c_I,h_I>0\) such that
		\begin{equation}
		\{\operatorname{Re}A_{a,h}^0,\operatorname{Im}A_{a,h}^0\}
		+\rho^{-1}|A_{a,h}^0|^2
		\geq c_I\rho^{-1}
		\label{eq:theorem-B-positive-a}
		\end{equation}
		
		for every
		\begin{equation*}
		a\in I,\quad 0<h\leq h_I,\quad \rho\geq B_I,\quad
		x\in\Omega_\rho^*,\quad \xi\in\T^2.
		\end{equation*}
		There are \(B,c>0\) such that
		\begin{equation}
		\{\operatorname{Re}A_{0,h}^0,\operatorname{Im}A_{0,h}^0\}
		+2\left(
		\rho^{-2}|\operatorname{Re}A_{0,h}^0|^2
		+|\operatorname{Im}A_{0,h}^0|^2
		\right)
		\geq c\rho^{-2}
		\label{eq:theorem-B-zero}
		\end{equation}
		for every
		\begin{equation*}
		0<h\leq1,\quad \rho\geq B,\quad
		x\in\Omega_\rho^*,\quad \xi\in\T^2.
		\end{equation*}
	\end{proposition}
	Before proving Proposition~\ref{thm:symbol-positivity}, we need some
	preliminary results.
	Everything in this section concerns the comparison symbol
	\eqref{eq:comparison-symbol}.
	{By the definition of \(p_\lambda\)} and \eqref{eq:comparison-symbol},
	{\begin{equation*}
		A_{a,h}^0(x,\xi)
		=p_\lambda\bigl(\xi+i\nabla\phi_{a,h}(x)\bigr).
		\end{equation*}
	}

	For \(\omega=(\omega_1,\omega_2)\in\mathbb S^1\), set
	\begin{equation*}
	\omega^\perp=(-\omega_2,\omega_1).
	\end{equation*}
	
	\begin{lemma}\label{lem:tangential-rigorous}
		Let \(I\Subset(0,a_\lambda)\).  There is \(c_I>0\) such that
		\begin{equation}
		|\omega^\perp\cdot\nabla p_\lambda(\xi+ib\omega)|^2
		\geq c_I
		\label{eq:transverse-gradient-bound}
		\end{equation}
		whenever
		\begin{equation*}
		p_\lambda(\xi+ib\omega)=0,
		\qquad \xi\in\T^2,
		\qquad \omega\in\mathbb S^1,
		\qquad b\in I.
		\end{equation*}
	\end{lemma}
	
	\begin{proof}
		Suppose that, for some $b\in I$, $\omega\in\mathbb S^1$, and
		$\xi\in\T^2$,
		\(\omega^\perp\cdot\nabla p_\lambda(\xi+ib\omega)=0\) and
		\(p_\lambda(\xi+ib\omega)=0\).  Then
		\begin{align}
		2\sum_{j=1}^2\left(
		\cos\xi_j\cosh(b\omega_j)
		-i\sin\xi_j\sinh(b\omega_j)\right)-\lambda&=0,
		\label{eq:characteristic-vanishing}\\
		\omega_2\left(
		\sin\xi_1\cosh(b\omega_1)
		+i\cos\xi_1\sinh(b\omega_1)\right)
		-\omega_1\left(
		\sin\xi_2\cosh(b\omega_2)
		+i\cos\xi_2\sinh(b\omega_2)\right)&=0.
		\label{eq:transverse-vanishing}
		\end{align}
		The imaginary part of \eqref{eq:characteristic-vanishing} and the
		real part of \eqref{eq:transverse-vanishing} give
		\begin{equation*}
		\begin{pmatrix}
		\tanh(b\omega_1)&\tanh(b\omega_2)\\
		\omega_2&-\omega_1
		\end{pmatrix}
		\begin{pmatrix}
		\sin\xi_1\cosh(b\omega_1)\\
		\sin\xi_2\cosh(b\omega_2)
		\end{pmatrix}
		=0.
		\end{equation*}
		The determinant of the matrix is
		\begin{equation*}
		-\omega_1\tanh(b\omega_1)
		-\omega_2\tanh(b\omega_2)\neq0,
		\end{equation*}
		because \(b>0\), \(|\omega|=1\), and
		\(\omega_j\tanh(b\omega_j)\geq0\).  Thus
		\(\sin\xi_1=\sin\xi_2=0\), so
		\(\cos\xi_j\in\{-1,1\}\).  Now the real part of
		\eqref{eq:characteristic-vanishing}  becomes
		\begin{equation*}
		\cos \xi_1\cosh(b\omega_1)+
		\cos \xi_2\cosh(b\omega_2)=\frac\lambda2.
		\end{equation*}
		Because \(0<\lambda/2<2\),   $\cos\xi_1$ and
		$\cos\xi_2$ have opposite signs.  Without loss of generality, assume that
		$\cos\xi_1=1$ and $\cos\xi_2=-1$.  Then
		\begin{equation*}
		\cosh(b\omega_1)
		=\cosh(b\omega_2)+\frac\lambda2
		\geq1+\frac\lambda2.
		\end{equation*}
		Therefore
		\begin{equation*}
		b\geq|b\omega_1|
		\geq\operatorname{arcosh}\left(1+\frac\lambda2\right)=a_\lambda,
		\end{equation*}
		contrary to \(b\in\overline I\).  Hence the continuous function
		\begin{equation*}
		|\omega^\perp\cdot\nabla p_\lambda(\xi+ib\omega)|^2
		\end{equation*}
		is strictly positive on the compact set where
		$p_\lambda(\xi+ib\omega)=0$.  Compactness now yields
		\eqref{eq:transverse-gradient-bound}.
	\end{proof}

	\begin{proof}[\bf Proof of Proposition~\ref{thm:symbol-positivity}]
		Using $A_{a,h}^0=p_\lambda\bigl(\xi+i\nabla\phi_{a,h}(x)\bigr)$, set
		\begin{equation*}
		q_{a,h}=\operatorname{Re}A_{a,h}^0,
		\qquad
		s_{a,h}=\operatorname{Im}A_{a,h}^0,
		\qquad
		w=\nabla p_\lambda\bigl(\xi+i\nabla\phi_{a,h}(x)\bigr).
		\end{equation*}
		Put $H=D^2\phi_{a,h}$.  The chain rule gives
		\[
		\partial_\xi A_{a,h}^0=w,
		\qquad
		\partial_xA_{a,h}^0=iHw.
		\]
		Therefore, by the definition of the Poisson bracket,
		\begin{align}
		\{q_{a,h},s_{a,h}\}
		&=\operatorname{Re}w\cdot\operatorname{Im}(iHw)
		-\operatorname{Re}(iHw)\cdot\operatorname{Im}w\notag\\
		&=(\operatorname{Re}w)^TH\operatorname{Re}w
		+(\operatorname{Im}w)^TH\operatorname{Im}w\notag\\
		&=\overline w^{\,T}D^2\phi_{a,h}w.
		\label{eq:poisson-bracket-identity}
		\end{align}
		For \(x\in\Omega_\rho^*\) and \(\rho\) sufficiently large,
		we have \(x\neq0\).  Write
		\begin{equation*}
		R=R_h(x),
		\qquad
		\omega=\frac{x}{|x|}.
		\end{equation*}
		Then, uniformly for \(0<h\leq1\),
		\begin{equation*}
		\frac{x}{R}\longrightarrow\omega
		\qquad\text{as }\rho\longrightarrow\infty.
		\end{equation*}
		For \(x=(x_1,x_2)\), let
		\begin{equation*}
		x\otimes x
		=\begin{pmatrix}
		x_1^2&x_1x_2\\
		x_1x_2&x_2^2
		\end{pmatrix}.
		\end{equation*}
		Direct differentiation gives
		\begin{align}\label{posa}
		\nabla\phi_{a,h}
		&=\left(a+R^{-1}\right)\frac{x}{R},\\
		R D^2\phi_{a,h}
		&=a\left(\operatorname{Id}-\frac{x\otimes x}{R^2}\right)
		+R^{-1}\left(\operatorname{Id}-2\frac{x\otimes x}{R^2}\right).
		\end{align}
		Indeed, since $x/R\to\omega$, the formulas above give
		\begin{equation*}
		\nabla\phi_{a,h}\longrightarrow a\omega,
		\qquad
		R D^2\phi_{a,h}\longrightarrow
		a\left(\operatorname{Id}-\omega\otimes\omega\right).
		\end{equation*}
		Moreover,
		\begin{equation*}
		A_{a,h}^0(x,\xi)\longrightarrow
		p_\lambda(\xi+ia\omega).
		\end{equation*}
		Consequently, \eqref{eq:poisson-bracket-identity} gives, uniformly for
		\(a\in I\), \(x\in\Omega_\rho^*\),
		\(\xi\in\T^2\), and \(0<h\leq1\),
		\begin{equation}
		R\{q_{a,h},s_{a,h}\}
		\longrightarrow
		a|\omega^\perp\cdot
		\nabla p_\lambda(\xi+ia\omega)|^2
		\label{eq:positive-bracket-limit}
		\end{equation}
		as \(\rho\to\infty\).
		By Lemma~\ref{lem:tangential-rigorous}, the right-hand side of
		\eqref{eq:positive-bracket-limit} is uniformly positive on
		\begin{equation*}
		\left\{(a,\omega,\xi):
		a\in I, \omega\in\mathbb S^1,
		\ p_\lambda(\xi+ia\omega)=0\right\}.
		\end{equation*}
		The two nonnegative terms below do not vanish simultaneously, so a
		compactness argument gives \(c_I>0\) such that
		\begin{equation}\label{gaug2731}
		a|\omega^\perp\cdot
		\nabla p_\lambda(\xi+ia\omega)|^2
		+\frac14|p_\lambda(\xi+ia\omega)|^2
		\geq 2c_I.
		\end{equation}
		Since \(1/4\leq R/\rho\leq8\) on \(\Omega_\rho^*\), \eqref{eq:positive-bracket-limit}  and \eqref{gaug2731}  imply, for sufficiently large \(\rho\),
		\begin{equation*}
		R\left(\{q_{a,h},s_{a,h}\}
		+\rho^{-1}|A_{a,h}^0|^2\right)
		=R\{q_{a,h},s_{a,h}\}
		+\frac R\rho|A_{a,h}^0|^2\geq c_I.
		\end{equation*}
		This proves
		\eqref{eq:theorem-B-positive-a}.
		
		For \(a=0\), note that  in \eqref{posa} the leading order vanishes. However, the same calculation gives
		\begin{equation*}
		\nabla\phi_{0,h}=\frac{x}{R^2},
		\qquad
		R^2D^2\phi_{0,h}=\operatorname{Id}
		-2\frac{x\otimes x}{R^2}.
		\end{equation*}
		Hence, by \eqref{eq:poisson-bracket-identity} and  $A_{0,h}^0=p_\lambda\bigl(\xi+i\nabla\phi_{0,h}(x)\bigr)$, uniformly for
		\(x\in\Omega_\rho^*\), \(\xi\in\T^2\),
		and \(0<h\leq1\),
		\begin{align*}
		R^2\{q_{0,h},s_{0,h}\}
		&\longrightarrow
		|\nabla p_\lambda(\xi)|^2
		-2|\nabla p_\lambda(\xi)\cdot\omega|^2,\\
		q_{0,h}&\longrightarrow p_\lambda(\xi),\\
		Rs_{0,h}&\longrightarrow
		\nabla p_\lambda(\xi)\cdot\omega.
		\end{align*}
		
		Consequently,
		\begin{align*}
		&R^2\left[
		\{q_{0,h},s_{0,h}\}
		+2\left(\rho^{-2}|q_{0,h}|^2+|s_{0,h}|^2\right)
		\right]\\
		&\qquad=
		|\nabla p_\lambda(\xi)|^2
		+2\frac{R^2}{\rho^2}|p_\lambda(\xi)|^2+o(1)
		\end{align*}
		uniformly as \(\rho\to\infty\).  Since
		\(1/4\leq R/\rho\leq8\), the expression on the right is bounded
		below by
		\[
		|\nabla p_\lambda(\xi)|^2
		+\frac18|p_\lambda(\xi)|^2+o(1).
		\]
		This is uniformly positive for sufficiently large $\rho$, because
		\(p_\lambda\) and \(\nabla p_\lambda\) do not vanish simultaneously. So
		\begin{equation}
		R^2\left[
		\{q_{0,h},s_{0,h}\}
		+2\left(\rho^{-2}|q_{0,h}|^2+|s_{0,h}|^2\right)
		\right]\geq c.
		\end{equation}
		This  proves
		\eqref{eq:theorem-B-zero}.
	\end{proof}

	\subsection{Proof of Theorem~\ref{thm:carleman-estimates}}
	Choose \(\bar a\in(0,a_\lambda)\) such that \(I\subset(0,\bar a]\), and use the
	notation of Lemma~\ref{lem:operator-comparison}.  All constants depending
	on \(\bar a\) are henceforth absorbed into constants carrying the
	subscript \(I\).
	
	\medskip
	\noindent\emph{{The case \(a\in I\)}.}
	{The nonnegative symbol (by Proposition \ref{thm:symbol-positivity})
		\begin{equation*}
		b=\{q,s\}+\rho^{-1}(q^2+s^2)-c_I\rho^{-1}
		\end{equation*}
		has fixed finite Fourier support which is immediate from the definitions \eqref{symbol12}.  }Proposition~\ref{prop:comparison-coefficient-estimates}
	and the product rule show that all its Fourier coefficients satisfy
	\begin{equation*}
	|\partial_x^\beta\widehat b(x,g)|
	\leq C_I\rho^{-1-|\beta|}
	\qquad \text{where }\quad |\beta|\leq1.
	\end{equation*}
	Hence Lemma~\ref{lem:friedrichs},
	Proposition~\ref{thm:symbol-positivity},
	\eqref{eq:commutator-operator-comparison}, and
	\eqref{eq:quadratic-operator-comparison} give
	\begin{align*}
	c_Ih\rho^{-1}\norm v_2^2
	&\leq h\operatorname{Re}\ip{
		\operatorname{Op}_h(\{q,s\})v}{v}
	+h\rho^{-1}\operatorname{Re}\ip{
		\operatorname{Op}_h(q^2+s^2)v}{v}
	+C_I\left(h\rho^{-2}+h^3\rho^{-1}\right)\norm v_2^2\\
	&\leq \ip{i[\mathcal Q,\mathcal S]v}{v}
	+h\rho^{-1}
	\left(\norm{\mathcal Qv}_2^2+\norm{\mathcal Sv}_2^2\right)
	+C_I\left(h\rho^{-2}+h^3\rho^{-1}\right)\norm v_2^2\\
	&=\left(\frac12+\frac12h\rho^{-1}\right)
	\norm{\mathcal Av}_2^2
	-\left(\frac12-\frac12h\rho^{-1}\right)
	\norm{\mathcal A^*v}_2^2
	+C_I\left(h\rho^{-2}+h^3\rho^{-1}\right)\norm v_2^2\\
	&\leq\norm{\mathcal Av}_2^2
	+C_I\left(h\rho^{-2}+h^3\rho^{-1}\right)\norm v_2^2,
	\end{align*}
	{		where the equality follows from}
	\eqref{eq:quadratic-exact-identity} and the last inequality follows since $h \rho^{-1}$ is small.  Absorbing the last term gives
	\begin{equation}
	h\rho^{-1}\norm v_2^2
	\leq C_I\norm{\mathcal A_{a,h}^0v}_2^2.
	\label{eq:comparison-Carleman-positive-a}
	\end{equation}
	By Lemma~\ref{lem:comparison-symbols} and
	\eqref{eq:comparison-Carleman-positive-a}, we obtain
	\eqref{eq:theorem-A-positive-a}.
	
	\medskip
	\noindent\emph{The case \(a=0\).}
	The proof is the same, using \eqref{eq:theorem-B-zero},
	Lemma~\ref{lem:friedrichs} with \(\alpha=2\),
	\eqref{eq:quadratic-operator-comparison-zero}, and
	\eqref{eq:comparison-operator-error-zero} in place of their
	positive-\(a\) counterparts.  Since the constants in these estimates are
	independent of \(I\) and \(\bar a\), the resulting \(B,C,h_0\) are also
	independent of \(I\).  This gives \eqref{eq:theorem-A-zero}.
	\qed

	\subsection{Weighted estimates}\label{sec:weighted-estimates}
	
	As consequences of Theorem~\ref{thm:carleman-estimates}, we obtain the
	following two weighted estimates.
	
	\begin{proposition}\label{prop:polynomial-rigorous}
		There are constants \(C,B>0\)   and an integer  \(k_0\geq1\) such that for every integer \(k\geq k_0\) 
		and every finitely supported \(f\),
		\begin{equation*}
		\norm{\la n\ra^k f}_2
		\leq C\norm{\la n\ra^{k+1}P_\lambda f}_2
		+(Bk)^k\norm f_2.
		\end{equation*}
	\end{proposition}
	
	\begin{proposition}
		\label{prop:exponential-rigorous}
		Let \(I\Subset(0,a_\lambda)\).  There are constants \(C_I,B_I>0\) and
		an integer \(k_I\geq1\) such that, for every \(a\in I\), every integer
		\(k\geq k_I\), and every finitely supported \(f\),
		\begin{equation*}
		\norm{e^{a\la n\ra}\la n\ra^k f}_2
		\leq C_I\norm{e^{a\la n\ra}\la n\ra^{k+1}P_\lambda f}_2
		+(B_Ik)^k\norm{e^{a\la n\ra}f}_2.
		\end{equation*}
	\end{proposition}

	Towards proving the above, we first present a preliminary
	lemma (recall the definition of $R_h$ from \eqref{norm12}).

	\begin{lemma}\label{lem:exterior-carleman}
		Let \(I\Subset(0,a_\lambda)\).  There are \(B_I\geq1\) and \(C_I,h_I>0\)
		such that, for \(a\in I\), \(0<h\leq h_I\), and every finitely supported
		\(w\) satisfying
		\(\operatorname{supp}w\subset\{R_h\geq B_I\}\), one has
		\begin{equation}
		h^{1/2}\norm w_2
		\leq C_I\norm{R_h\mathcal A_{a,h}w}_2.
		\label{eq:exterior-carleman}
		\end{equation}
		
		There are \(B\geq1\) and \(C,h_0>0\), independent of \(I\), such that,
		for \(0<h\leq h_0\) and every finitely supported \(w\) satisfying
		\(\operatorname{supp}w\subset\{R_h\geq B\}\), one has  
		\begin{equation}
		h^{1/2}\norm w_2
		\leq C\norm{R_h\mathcal A_{0,h}w}_2.
		\label{eq:exterior-carleman-zero}
		\end{equation}
		where on the right hand side $R_h$ acts multiplicatively.
	\end{lemma}
	
	\begin{proof}
		We prove \eqref{eq:exterior-carleman} and
		\eqref{eq:exterior-carleman-zero} simultaneously.  Throughout the proof,
		\(B\) and \(C\) denote the constants for the case under consideration;
		for \eqref{eq:exterior-carleman} they may depend on \(I\), whereas for
		\eqref{eq:exterior-carleman-zero} they are independent of \(I\).
		Set
		\begin{equation}\label{g4}
		\rho_j=2^jB,
		\qquad j=0,1,2,\ldots.
		\end{equation}
		Choose \(\{\psi_j\}_{j\geq0}\) such that
		\begin{equation}\label{g3}
		\sum_{j\geq0}\psi_j^2=1
		\quad\text{on }\{R_h\geq B\},
		\qquad
		\operatorname{supp}\psi_j\subset\Omega_{\rho_j},
		\qquad
		|\partial_x^\beta\psi_j|\leq C\rho_j^{-|\beta|},
		\quad |\beta|=0,1.
		\end{equation}
		By \eqref{eq:exact-conjugation},
		\begin{equation}
		[\mathcal A_{a,h},\psi_j]w(x)
		=\sum_{|e|_1=1}e^{-d_e(x)}
		\bigl(\psi_j(x+he)-\psi_j(x)\bigr)w(x+he).
		\label{eq:annular-commutator-identity}
		\end{equation}
		By \eqref{eq:weight-derivative-bounds}, for $x$ with $R_h(x)\geq B$, 
		\begin{equation}\label{g1}
		|\nabla\phi_{a,h}(x)|
		\leq C.
		\end{equation}
		For $x$ with $R_h(x)\geq B$, \eqref{gde} and \eqref{g1} imply
		\begin{equation}
		|d_e(x)|\leq C.
		\label{eq:annular-exponential-factor}
		\end{equation}
		By \eqref{g3},
		\begin{equation}
		\left|
		\bigl(\psi_j(x+he)-\psi_j(x)\bigr)w(x+he)
		\right|
		\leq Ch\rho_j^{-1}
		\mathbf 1_{\Omega_{\rho_j}^*}(x+he)|w(x+he)|.
		\label{eq:annular-cutoff-factor}
		\end{equation}
		By \eqref{eq:annular-commutator-identity},
		\eqref{eq:annular-exponential-factor}, and
		\eqref{eq:annular-cutoff-factor}, we obtain
		\begin{equation}
		\norm{[\mathcal A_{a,h},\psi_j]w}_2
		\leq Ch\rho_j^{-1}
		\norm{\mathbf 1_{\Omega_{\rho_j}^*}w}_2.
		\label{eq:annular-commutator}
		\end{equation}
		By \eqref{g3} and the support assumption on \(w\),
		\begin{equation}
		\sum_j\norm{\psi_jw}_2^2=\norm w_2^2.
		\label{eq:dyadic-partition-norm}
		\end{equation}
		Let \(a\in I\).  By \eqref{eq:theorem-A-positive-a},
		\eqref{eq:annular-commutator}, \eqref{eq:dyadic-partition-norm}, and
		\begin{equation*}
		\mathcal A_{a,h}(\psi_jw)
		=\psi_j\mathcal A_{a,h}w+[\mathcal A_{a,h},\psi_j]w,
		\end{equation*}
		we obtain
		\begin{equation}
		h\norm w_2^2
		\leq
		C\sum_j\rho_j\norm{\psi_j\mathcal A_{a,h}w}_2^2
		+Ch^2\sum_j\rho_j^{-1}
		\norm{\mathbf 1_{\Omega_{\rho_j}^*}w}_2^2.
		\label{eq:positive-a-dyadic}
		\end{equation}
		By \eqref{g4},
		\begin{align}
		\sum_j\rho_j^{-1}
		\norm{\mathbf 1_{\Omega_{\rho_j}^*}w}_2^2
		&\leq\norm w_2^2\sum_{j\geq0}\rho_j^{-1} \notag\\
		&\leq CB^{-1}\norm w_2^2.
		\label{eq:positive-a-overlap}
		\end{align}
		By \eqref{eq:positive-a-dyadic},
		\eqref{eq:positive-a-overlap}, and
		\(ChB^{-1}\leq1/2\),
		\begin{equation}
		\frac h2\norm w_2^2
		\leq C\sum_j\rho_j
		\norm{\psi_j\mathcal A_{a,h}w}_2^2.
		\label{eq:positive-a-dyadic-final}
		\end{equation}
		Let \(a=0\).  By \eqref{eq:theorem-A-zero},
		\eqref{eq:annular-commutator}, and
		\eqref{eq:dyadic-partition-norm},
		\begin{equation}
		h\norm w_2^2
		\leq
		C\sum_j\rho_j^2\norm{\psi_j\mathcal A_{0,h}w}_2^2
		+Ch^2\sum_j
		\norm{\mathbf 1_{\Omega_{\rho_j}^*}w}_2^2.
		\label{eq:zero-a-dyadic}
		\end{equation}
		Since each $x$ belongs to only finitely many sets
		\(\Omega_{\rho_j}^*\),
		\begin{equation}
		\sum_j\norm{\mathbf 1_{\Omega_{\rho_j}^*}w}_2^2
		\leq C\norm w_2^2.
		\label{eq:zero-a-overlap}
		\end{equation}
		By \eqref{eq:zero-a-dyadic}, \eqref{eq:zero-a-overlap}, and
		\(Ch\leq1/2\),
		\begin{equation}
		\frac h2\norm w_2^2
		\leq C\sum_j\rho_j^2
		\norm{\psi_j\mathcal A_{0,h}w}_2^2.
		\label{eq:zero-a-dyadic-final}
		\end{equation}
		If \(\psi_j(x)\neq0\), then \(x\in\Omega_{\rho_j}\), and
		\begin{equation}\label{g6}
		\rho_j^2|\psi_j(x)|^2
		\leq CR_h(x)^2|\psi_j(x)|^2.
		\end{equation}
		By \eqref{g3}, for every $x$,
		\begin{equation}\label{g7}
		\sum_{j}|\psi_j(x)|^2\leq C.
		\end{equation}
		By \eqref{g6} and \eqref{g7},
		\begin{equation}
		\sum_j\rho_j^\ell|\psi_j(x)|^2
		\leq CR_h(x)^2,
		\qquad \ell=1,2.
		\label{eq:weighted-dyadic-partition}
		\end{equation}
		By \eqref{eq:weighted-dyadic-partition},
		\begin{align}
		\sum_j\rho_j^\ell\norm{\psi_j \mathcal A_{a,h}w}_2^2
		&=\sum_{x\in\Lambda_h}| \mathcal A_{a,h}w(x)|^2
		\sum_j\rho_j^\ell|\psi_j(x)|^2\\
		&\leq C\sum_{x\in\Lambda_h}R_h(x)^2|\mathcal A_{a,h}w(x)|^2\\
		&=C\norm{R_h \mathcal A_{a,h}w}_2^2,
		\qquad \ell=1,2.\label{g8}
		\end{align}
		By \eqref{eq:positive-a-dyadic-final},
		\eqref{eq:zero-a-dyadic-final}, and \eqref{g8},
		\begin{equation*}
		h\norm w_2^2
		\leq C\norm{R_h\mathcal A_{a,h}w}_2^2.
		\end{equation*}
		This proves \eqref{eq:exterior-carleman} and
		\eqref{eq:exterior-carleman-zero}.
	\end{proof}
	
	\begin{proof}[\bf Proofs of Propositions~\ref{prop:polynomial-rigorous} and
		\ref{prop:exponential-rigorous}]
		We prove the cases \(a\in I\) and \(a=0\) simultaneously.  Throughout
		this proof, \(B\) and \(C\) denote the constants for the case under
		consideration; when \(a\neq0\), they may depend on \(I\), whereas when
		\(a=0\), they are independent of \(I\).

		Set \(h=k^{-1}\),
		regard \(f\) as a function on \(\Lambda_h\) by
		setting \(f(hn)=f(n)\), and put
		\begin{equation*}
		u=e^{\phi_{a,h}/h}f.
		\end{equation*}
		Choose \(\chi_0\in C^\infty([0,\infty);[0,1])\) such that
		\begin{equation*}
		\chi_0(t)=0\quad\text{for }t\leq 3,
		\qquad
		\chi_0(t)=1\quad\text{for }t\geq 4,
		\end{equation*}
		and define
		\begin{equation*}
		\chi(x)=\chi_0\left(\frac{R_h(x)}{B}\right).
		\end{equation*}
		Then \(\operatorname{supp}(\chi u)\subset\{R_h\geq B\}\), so
		Lemma~\ref{lem:exterior-carleman} applies to \(\chi u\).\\
		
		By the definition of \(\mathcal A_{a,h}\),
		\begin{equation*}
		\mathcal A_{a,h}u=e^{\phi_{a,h}/h}P_\lambda f.
		\end{equation*}
		Hence
		\begin{equation*}
		\mathcal A_{a,h}(\chi u)
		=\chi e^{\phi_{a,h}/h}P_\lambda f
		+[\mathcal A_{a,h},\chi]u.
		\end{equation*}
		
		{We first estimate the cutoff commutator}.  By
		\eqref{eq:exact-conjugation},
		\begin{equation}
		(R_h[\mathcal A_{a,h},\chi]u)(x)
		=R_h(x)\sum_{|e|_1=1}e^{-d_e(x)}
		\bigl(\chi(x+he)-\chi(x)\bigr)u(x+he).
		\label{eq:cutoff-commutator-identity}
		\end{equation}
		Since \(|\nabla R_h|\leq1\),
		\begin{equation}
		|R_h(x+he)-R_h(x)|\leq h,
		\qquad
		|\chi(x+he)-\chi(x)|\leq ChB^{-1}.
		\label{eq:cutoff-difference}
		\end{equation}
		If $R_h(x)\geq5B$ or $R_h(x)\leq2B$, then
		\eqref{eq:cutoff-difference} gives
		\begin{equation}
		\bigl(\chi(x+he)-\chi(x)\bigr)=0.
		\end{equation}
		Therefore, in \eqref{eq:cutoff-commutator-identity}, we need only
		consider $x$ such that $2B\leq R_h(x)\leq5B$.
		
		Thus \eqref{eq:cutoff-commutator-identity},
		\eqref{eq:cutoff-difference}, and
		\eqref{eq:annular-exponential-factor} give
		\begin{equation}
		\norm{R_h[\mathcal A_{a,h},\chi]u}_2
		\leq Ch
		\norm{\mathbf 1_{\{B\leq R_h\leq 6B\}}u}_2.
		\label{eq:cutoff-commutator}
		\end{equation}

		By \eqref{eq:exterior-carleman}--\eqref{eq:exterior-carleman-zero} and
		\eqref{eq:cutoff-commutator},
		\begin{equation}
		h^{1/2}\norm{\chi u}_2
		\leq C\norm{R_he^{\phi_{a,h}/h}P_\lambda f}_2
		+Ch
		\norm{\mathbf 1_{\{B\leq R_h\leq 6B\}}u}_2.
		\label{eq:cutoff-exterior-estimate}
		\end{equation}
		Since $\chi(x)=1 $ for $R_h(x)\geq 4B$, \eqref{eq:cutoff-exterior-estimate} implies
		\begin{equation}
		h^{1/2}\norm{u}_2
		\leq C\norm{R_he^{\phi_{a,h}/h}P_\lambda f}_2
		+Ch
		\norm{\mathbf 1_{\{B\leq R_h\leq 6B\}}u}_2+h^{1/2}\norm{\mathbf 1_{\{ R_h\leq 4B\}}u}_2.
		\label{eq:cutoff-exterior-estimate1}
		\end{equation}
		At \(x=hn\), \eqref{eq:scaled-radius} and
		\eqref{eq:weight-scaling} give
		\begin{align}
		u(hn)
		&=(e^{\phi_{a,h}/h}f)(hn)=h^ke^{a\la n\ra}\la n\ra^kf(n), \label{crudebound}\\
		R_h(hn)e^{\phi_{a,h}(hn)/h}(P_\lambda f)(n)
		&=h^{k+1}e^{a\la n\ra}\la n\ra^{k+1}(P_\lambda f)(n).
		\label{eq:cutoff-scaling-identities}
		\end{align}
		Thus, by \eqref{eq:scaled-exterior-equivalence} and  \eqref{crudebound}
		\begin{align}
		\norm{\mathbf 1_{ R_h(x)\leq 6B\}}u}_2
		&\leq h^k(Bk)^k\norm{e^{a\la n\ra}f}_2.
		\label{eq:cutoff-interior-bounds}
		\end{align}

		Equations \eqref{eq:cutoff-exterior-estimate1}--
		\eqref{eq:cutoff-interior-bounds} give
		\begin{align*}
		h^{k+1/2}\norm{e^{a\la n\ra}\la n\ra^kf}_2
		&\leq Ch^{k+1}
		\norm{e^{a\la n\ra}\la n\ra^{k+1}P_\lambda f}_2+h^{k+1/2}(Bk)^k
		\norm{e^{a\la n\ra}f}_2.
		\end{align*}
		Dividing by \(h^{k+1/2}\) proves
		Proposition~\ref{prop:exponential-rigorous}. 
		Taking \(a=0\) in the same argument proves
		Proposition~\ref{prop:polynomial-rigorous}.

	\end{proof}
	
	\section{Absence of bulk, non-critical eigenvalues}\label{sec:exponential-decay}
	The main goal of this section is to prove
	Theorem~\ref{thm1:absence}.  Accordingly, throughout the remainder of the
	paper, we assume that $0<\lambda<4$,
	\[
	|V(n)|\leq C\la n\ra^{-\kappa},
	\qquad \kappa>1,
	\]
	and that $u\in\ell^2(\Z^2)$ satisfies
	$(\Delta+V)u=\lambda u$.  We set \(\tau=\kappa-1>0\).
	
	\subsection{Exponentially decaying eigensolutions}
	In this subsection, we prove that every such $\ell^2$ solution decays
	exponentially. We start with a preliminary lemma. 
	\begin{lemma}\label{lem:fixed-division-note}
		For every real \(s\geq1\), there is \(C_{s,\lambda}>0\) such that,
		whenever \(g\in L^2(\T^2)\) and
		\(p_\lambda g\in H^s(\T^2)\), one has
		\(g\in H^{s-1}(\T^2)\) and
		\begin{equation}\label{eq:real-order-division}
		\norm{g}_{H^{s-1}(\T^2)}
		\leq C_{s,\lambda}\norm{p_\lambda g}_{H^s(\T^2)}.
		\end{equation}
	\end{lemma}
	
	\begin{proof}
		Put \(f=p_\lambda g\) and
		\[
		\Sigma_\lambda=\{\xi\in\T^2:p_\lambda(\xi)=0\}.
		\]
		Since \(0<\lambda<4\), the gradient of \(p_\lambda\) does not vanish on
		\(\Sigma_\lambda\).  Thus, near each point of \(\Sigma_\lambda\), there
		are smooth coordinates \((t,y)\) on a ball
		\(B_\varepsilon\subset\R^2\) centered at the origin such that
		\(t=p_\lambda(\xi)\).  {After inserting a smooth cutoff, write
			\[
			U=(\chi g)\circ\Phi^{-1},
			\qquad
			F=(\chi f)\circ\Phi^{-1}.
			\]}
		Then \(F=tU\), with \(U\in L^2(B_\varepsilon)\) and
		\(F\in H^s(B_\varepsilon)\).  We prove the required estimate in these
		coordinates.
		
		For a smooth function \(F\) on \(B_\varepsilon\), define
		\[
		(\mathcal DF)(t,y)
		=\int_0^1\partial_tF(\theta t,y)\,d\theta.
		\]
		Let \(m\geq1\) be an integer.  If \(j,k\geq0\) and
		\(j+k\leq m-1\), then
		\[
		\partial_t^j\partial_y^k\mathcal DF(t,y)
		=\int_0^1\theta^j
		(\partial_t^{j+1}\partial_y^kF)(\theta t,y)\,d\theta.
		\]
		The ball \(B_\varepsilon\) is invariant under
		\((t,y)\mapsto(\theta t,y)\) for \(0<\theta\leq1\), and a change of
		variable in \(t\) gives
		\[
		\norm{G(\theta\,\cdot,\cdot)}_{L^2(B_\varepsilon)}
		\leq\theta^{-1/2}\norm G_{L^2(B_\varepsilon)}.
		\]
		Applying this estimate with
		\(G=\partial_t^{j+1}\partial_y^kF\), Minkowski's inequality gives
		\begin{align*}
		\norm{\partial_t^j\partial_y^k\mathcal DF}_{L^2(B_\varepsilon)}
		&\leq\int_0^1\theta^j
		\norm{G(\theta\,\cdot,\cdot)}_{L^2(B_\varepsilon)}\,d\theta\\
		&\leq\left(\int_0^1\theta^{j-1/2}\,d\theta\right)
		\norm G_{L^2(B_\varepsilon)}\\
		&=\frac{1}{j+\frac12}
		\norm{\partial_t^{j+1}\partial_y^kF}_{L^2(B_\varepsilon)}.
		\end{align*}
		This implies 
		\[
		\norm{\mathcal DF}_{H^{m-1}(B_\varepsilon)}
		\leq C_m\norm F_{H^m(B_\varepsilon)}.
		\]
		and hence
		\[
		\norm{\mathcal DF}_{H^{m}(B_\varepsilon)}
		\leq C_{m+1}\norm F_{H^{m+1}(B_\varepsilon)}.
		\]
		The interpolation leads to
		\begin{equation}\label{eq:local-real-order-division}
		\norm{\mathcal DF}_{H^{s-1}(B_\varepsilon)}
		\leq C_s\norm F_{H^s(B_\varepsilon)}.
		\end{equation}
		
		It remains to identify \(\mathcal DF\).  Since \(F\in H^1(B_\varepsilon)\),
		Fubini's theorem shows that \(F(\,\cdot\,,y)\in H^1\) for almost every
		\(y\).  If \(F(0,y)\neq0\), then \(U(t,y)=F(t,y)/t\) is not square
		integrable near \(t=0\).  Therefore \(F(0,y)=0\) for almost every \(y\),
		\[
		\mathcal DF(t,y)
		=U(t,y).
		\]
		
		The local estimate now follows from
		\eqref{eq:local-real-order-division}. 
		{Now using the compactness of $\T^2$ and $\Sigma_\lambda,$ we can cover the latter by small balls and apply the above local estimate.  {Smooth cutoffs and coordinate
				changes keep  the Sobolev norm bounded. Finally, we can sum over a finite partition
				of unity on the cover. Off the cover, 
				\(p_\lambda\) is away from zero and hence dividing by it does not change the Sobolev norm either up to constants. Together this proves
				\eqref{eq:real-order-division}.}}
	\end{proof}
	
	The next proposition establishes an exponential decay of the eigenfunction with any rate smaller than but \emph{not including} the value $a_\lambda.$
	
	\begin{proposition}\label{thm:subcritical-exponential-decay}
		Suppose that $V:\Z^2\to\mathbb C$ satisfies
		\[
		|V(n)|\leq C\la n\ra^{-\kappa}
		\]
		for some $\kappa>1$.  If
		\[
		(\Delta+V)u=\lambda u,
		\qquad
		u\in\ell^2(\Z^2),
		\] 
		then
		\[
		e^{a\la n\ra}u\in\ell^2(\Z^2)
		\qquad\text{for every }0<a<a_\lambda.
		\]
	\end{proposition}
	
	\begin{proof}

		Set
		\[
		M_\gamma=\norm{\la n\ra^\gamma u}_2,
		\qquad \gamma\geq0.
		\]
		Applying Lemma~\ref{lem:fixed-division-note} to the Fourier transform of
		$u$ gives
		$M_\gamma<\infty$ for every
		$\gamma\geq0$ (note that we are using the well known fact that for any function $f$  on $\Z^2$, $\norm{\widehat f}_{H^{s}(\T^2)}^2\simeq\sum_{n\in \Z^2}  \la n\ra^{2s} |f(n)|^2$).  Indeed, from $P_\lambda u=-Vu$ and
		$|V(n)|\leq C\la n\ra^{-\kappa}$, finiteness of $M_\gamma$ implies
		finiteness of $M_{\gamma+\kappa-1}$; starting with $M_0<\infty$ and
		iterating proves the claim. 		
		
			For every sufficiently large integer $k$, all norms on the right-hand side of Proposition~\ref{prop:polynomial-rigorous} are finite.  Note that the proposition applies a priori to functions with finite support. However for infinitely supported functions, if the right hand side norms are finite, one can take finite restrictions and then pass to the limit.  Thus the proposition applies to  $u$ as well.   Using this and  that $P_\lambda u=-Vu$, we obtain
		\begin{equation}\label{gaug2715}
		M_k\leq A M_{k-\tau}+(Bk)^kM_0.
		\end{equation}
		For large $k$, interpolation (which is an immediate consequence of H\"older's inequality) gives
		\begin{equation}\label{gaug2716}
		M_{k-\tau}\leq M_0^{\tau/k}M_k^{1-\tau/k}.
		\end{equation}
		By \eqref{gaug2715} and  \eqref{gaug2716},   there are
		constants $C>0$ and $D>0$ such that for large $k$,
		\[
		M_k\leq C(Dk)^kM_0.
		\]

		For every $n\in\Z^2$ and every integer $k\geq1$, it follows that
		\[
		|u(n)|
		\leq \la n\ra^{-k}M_k
		\leq C\left(\frac{Dk}{\la n\ra}\right)^kM_0.
		\]
		Choose a sufficiently small $\delta>0$ and take
		\[
		k=\lfloor \delta\la n\ra\rfloor.
		\]
		For $\la n\ra$ sufficiently large, this gives
		\[
		|u(n)|
		\leq Ce^{-\delta\la n\ra}.
		\]
		Assume that  for some $a\in I$,
		\[
		M_\gamma(a)=\norm{\la n\ra^\gamma e^{a\la n\ra}u}_2<\infty
		\qquad \text{ for all } \gamma\geq0.
		\]
By Proposition 	\ref{prop:exponential-rigorous} (with approximation)  and using $P_\lambda u=-Vu$, we obtain
		\[
		M_k(a)
		\leq A_I M_{k-\tau}(a)+(B_Ik)^kM_0(a).
		\]
		Similarly, we obtain
		\[
		M_k(a)\leq C(Dk)^kM_0(a),
		\]
		where $C$ and $D$ depend on $I$.
		Therefore, for large $k$,
		\[
		e^{a\la n\ra}|u(n)|
		\leq \la n\ra^{-k}M_k(a)
		\leq C\left(\frac{Dk}{\la n\ra}\right)^kM_0(a).
		\]
		Choose a sufficiently small $\delta>0$ and take
		\[
		k=\lfloor \delta\la n\ra\rfloor.
		\]
		For $\la n\ra$ sufficiently large, this gives
		\[
		e^{a\la n\ra}|u(n)|
		\leq C e^{-\delta\la n\ra}.
		\]
		Hence
		\[
		|u(n)|
		\leq C e^{-(a+\delta)\la n\ra}.
		\]
		We can therefore improve the pointwise exponential decay rate to
		$a+\delta$, where $\delta$ depends on $I$.  
		Iterating this process completes the proof.
	\end{proof}
	
	\subsection{Proof of Theorem~\ref{thm1:absence}}\label{sec:absence-noncritical}

	By Proposition~\ref{thm:subcritical-exponential-decay},
	$e^{a\la n\ra}u\in\ell^2(\Z^2)$ for every $0<a<a_\lambda$.
	The purpose of this section is to obtain an estimate at $a=a_\lambda$ and
	then use that estimate to prove that $u$ vanishes.
	Set
	\begin{equation}\label{eq:critical-weight}
	\rho=e^{a_\lambda}
	=1+\frac{\lambda}{2}+\frac12\sqrt{\lambda(\lambda+4)}.
	\end{equation}
	Since $\cosh a_\lambda=1+\frac{\lambda}{2}$, we have
	\begin{equation}\label{eq:critical-weight-identity}
	\rho+\rho^{-1}=2\cosh a_\lambda=2+\lambda.
	\end{equation}
	As in Section~\ref{sec:absence-zero}, write $n=(m,j)\in\Z^2$ and introduce
	\[
	U_m=(u(m,j))_{j\in\Z}\in\cH.
	\]
	Recall that $\cH=\ell^2(\Z)$ and that $S$ denotes the shift given by
	$(Sw)_j=w_{j+1}$.  Define the bounded operator
	\[
	A_\lambda=\lambda I-(S+S^*).
	\]
	For each $m\in\Z$, let $V_m$ be the multiplication operator on $\cH$
	given by
	\[
	(V_mw)_j=V(m,j)w_j.
	\]
	Set $H_m=-V_mU_m$.  The equation
	$(\Delta+V)u=\lambda u$ is then equivalent to
	\begin{equation}\label{eq:noncritical-slice-recurrence}
	U_{m+1}+U_{m-1}-A_\lambda U_m=H_m.
	\end{equation}
	We apply the partial Fourier transform $\mathcal F$ introduced in
	Subsection~\ref{subsec:partial-fourier}.
	Under this transform, $S+S^*$ becomes multiplication by
	$2\cos\theta$.  Consequently, $A_\lambda$ becomes multiplication by
	\[
	b(\theta)=\lambda-2\cos\theta.
	\]
	For $s>1$, define an operator on sequences
	$f=(f_m)_{m\in\Z}$ with $f_m\in L^2(\T)$ by
	\begin{equation}\label{eq:Ls-definition}
	(L_sf)_m(\theta)
	=s^{-1}f_{m+1}(\theta)+sf_{m-1}(\theta)
	-b(\theta)f_m(\theta).
	\end{equation}
	The definition is designed to be particularly compatible with  weight $s^m$. That is, if
	\[
	\widehat U_m=\mathcal FU_m,\quad \quad f_m=s^m\widehat U_m,
	\qquad \widehat H_m=\mathcal FH_m,
	\]
	then \eqref{eq:noncritical-slice-recurrence} gives
	\[
	(L_sf)_m=s^m\widehat H_m.
	\]
	
We next present a framework to study solutions of the above equation in appropriate Banach spaces.

	\subsection{A uniform inverse estimate}\label{subsec:uniform-inverse-estimate}
	
	Define the Banach spaces
	\[
	X=\ell^\infty(\Z;L^2(\T)),
	\qquad
	\|f\|_X=\sup_{m\in\Z}\|f_m\|_{L^2(\T)},
	\]
	and
	\[
	Y=\ell^1(\Z;L^2(\T)),
	\qquad
	\|h\|_Y=\sum_{m\in\Z}\|h_m\|_{L^2(\T)}.
	\]
	
	\begin{proposition}\label{thm:uniform-inverse}
		Let $s_0>1$.  There is a constant $C=C(s_0,\lambda)$ such that, for
		every $s\in[s_0,\rho]$ (where $\rho$ was defined in \eqref{eq:critical-weight}) and every $h\in Y$, the equation
		\[
		L_sf=h
		\]
		has a unique solution $f\in X$.  This solution satisfies
		\begin{equation}\label{eq:uniform-inverse-estimate}
		\sup_{m\in\Z}\|f_m\|_{L^2(\T)}
		\leq C\sum_{m\in\Z}\|h_m\|_{L^2(\T)}.
		\end{equation}
		At $s=\rho$, the solution also satisfies, for every $M\in\Z$,
		\begin{equation}\label{eq:one-sided-inverse-estimate}
		\sup_{m\geq M}\|f_m\|_{L^2(\T)}
		\leq C\sum_{\ell\geq M+1}\|h_\ell\|_{L^2(\T)}.
		\end{equation}
		
	\end{proposition}
	Before proceeding with the proof, let us comment that the  uniformity in the above statement is precisely what allows us to start with Proposition~\ref{thm:subcritical-exponential-decay} and obtain an effective estimate for $e^{a_\lambda \la n\ra}u$.\\

	To prove Proposition~\ref{thm:uniform-inverse}, we first establish the
	following lemma.
	
	\begin{lemma}[Uniqueness in $X$]\label{lem:uniqueness-X}
		Let $s>1$.  If $v\in X$ satisfies
		\[
		L_sv=0,
		\]
		then $v=0$.  Consequently, the equation $L_sf=h$ has at most one
		solution in $X$.
	\end{lemma}
	\begin{proof}
		The proof is similar to the uniqueness argument at the critical energy
		$\lambda=0$ in Subsection~\ref{subsec:partial-fourier}.
		We again solve the scalar recurrence using its characteristic roots.
		Assume $v\in X$ and $L_sv=0$.  By \eqref{eq:Ls-definition}, we have
		\begin{equation}\label{eq:Ls-definition1}
		s^{-1}v_{m+1}(\theta)+sv_{m-1}(\theta)
		-b(\theta)v_m(\theta)=0
		\end{equation}
		and
		\begin{equation}\label{eq:Ls-definition2}
		\sup_m \int _\T|v_m(\theta)|^2\frac{d\theta}{2\pi} \leq C.
		\end{equation}
		The recurrence \eqref{eq:Ls-definition1} shows that $v_m(\theta)$ is
		determined by $v_0(\theta)$ and $v_1(\theta)$.  Its characteristic
		equation is
		\[
		z^2-sb(\theta)z+s^2=0.
		\]

		When $-2<b(\theta)<2$, write $b(\theta)=2\cos\varphi$ and define
		\[
		z_1(\theta)=se^{i\varphi},
		\qquad z_2(\theta)=se^{-i\varphi}.
		\]
		When $b(\theta)>2$, write $b(\theta)=2\cosh t$ and define
		\[
		z_1(\theta)=se^t,
		\qquad z_2(\theta)=se^{-t}.
		\]
		
		Solving the recurrence \eqref{eq:Ls-definition1} directly gives
		\[
		v_m(\theta)
		=d_1(\theta)z_1(\theta)^m+d_2(\theta)z_2(\theta)^m,
		\]
		where
		\[
		d_1=\frac{v_1-z_2v_0}{z_1-z_2},
		\qquad
		d_2=\frac{z_1v_0-v_1}{z_1-z_2},
		\]
		and
		\begin{equation}\label{eq:uniqueness-coefficients}
		d_1 z_1^m=\frac{v_{m+1}-z_2v_m}{z_1-z_2},
		\qquad
		d_2 z_2^m=\frac{z_1v_m-v_{m+1}}{z_1-z_2}.
		\end{equation}
		Letting $m\to\infty$ in \eqref{eq:uniqueness-coefficients} and using
		\eqref{eq:Ls-definition2}, we find that
		$d_1(\theta)=d_2(\theta)=0$ whenever $-2<b(\theta)<2$.
		Similarly, letting $m\to\infty$ gives $d_1(\theta)=0$ whenever
		$|b(\theta)|>2$. Letting $m\to\pm\infty$ gives $d_2(\theta)=0$ whenever
		$|b(\theta)|>2$ and $z_2(\theta)\neq1$.
		Since $\{\theta\in\T:|b(\theta)|=2\text{ or }z_2(\theta)=1\}$ is
		finite, we conclude that $v=0$.
		\end{proof}
	
	\begin{proof}[\bf Proof of Proposition~\ref{thm:uniform-inverse}]
		Uniqueness follows from Lemma~\ref{lem:uniqueness-X}.
		Recall that \begin{equation*}
	(L_sf)_m(\theta)
	=s^{-1}f_{m+1}(\theta)+sf_{m-1}(\theta)
	-b(\theta)f_m(\theta).
	\end{equation*}
We now follow a natural ansatz to solve the problem using a suitable Green kernel. Namely,	for $h\in Y$, we present the following candidate solution for $L_sf=h$ using the Green kernel
		$G_s(\ell,\theta)$
		\begin{equation}\label{g29}
		f_m(\theta)=
		\sum_{\ell\in\Z}G_s(m-\ell,\theta)h_\ell(\theta).
		\end{equation}
Thus we must seek for $G_s(\ell,\theta)$ satisfying
		\begin{equation}\label{eq:green-kernel}
		s^{-1}G_s(\ell+1,\theta)+sG_s(\ell-1,\theta)
		-b(\theta)G_s(\ell,\theta)=\delta_0(\ell) \quad \text{for}\,\,
		\ell \in\Z.
		\end{equation}
		where $\delta_0$ is the Kronecker delta at the origin.

		 We now proceed to first solve for the kernel. 
		Away from
		$\ell=0$, we solve the homogeneous recurrence associated with
		\eqref{eq:green-kernel}: 
		\begin{equation}\label{eq:green-kernel1}
		s^{-1}G_s(\ell+1,\theta)+sG_s(\ell-1,\theta)
		-b(\theta)G_s(\ell,\theta)=0.
		\end{equation}
		As in the proof of Lemma~\ref{lem:uniqueness-X}, we first solve the
		characteristic equation
		\begin{equation}\label{g23}
		z^2-sb(\theta)z+s^2=0.
		\end{equation}
		On each half-line, we choose the bounded solution.  The remaining
		constants are then determined by the conditions at $\ell=0$ and
		$\ell=-1$ in
		\eqref{eq:green-kernel}.
		
		Because $0<\lambda<4$,
		\[
		b(\theta)\in[\lambda-2,\lambda+2]\subset(-2,\infty).
		\]
		
		We distinguish the following three cases.
		
		\medskip
		\noindent\emph{Case 1: $-2<b(\theta)<2$.}
		Choose $\varphi\in (0,\pi)$ so that
		\[
		b(\theta)=2\cos\varphi.
		\]
		The two roots of \eqref{g23} are
		\[
		z_1=se^{i\varphi},
		\qquad
		z_2=se^{-i\varphi},
		\]
		and hence $|z_1|=|z_2|=s>1$.  Thus $|z_1|^\ell$ and $|z_2|^\ell$
		grow exponentially as $\ell\to+\infty$, and we must choose the zero
		solution on
		$\ell\geq0$.  On $\ell\leq-1$, both $|z_1|^\ell$ and
		$|z_2|^\ell$ decay as $\ell\to-\infty$, so we write
		\[
		G_s(\ell,\theta)=c_1z_1^\ell+c_2z_2^\ell.
		\]
		First, matching
		with the zero solution on $\ell\geq0$ gives
		$G_s(0,\theta)=c_1+c_2=0$.  Second, setting $\ell=0$ in
		\eqref{eq:green-kernel} and using
		$G_s(0,\theta)=G_s(1,\theta)=0$ gives the condition
		$G_s(-1,\theta)=s^{-1}$.  Hence
		\[
		c_1+c_2=0,
		\qquad
		c_1z_1^{-1}+c_2z_2^{-1}=s^{-1},
		\]
		and therefore
		\[
		c_1=-\frac{1}{2i\sin\varphi},
		\qquad
		c_2=\frac{1}{2i\sin\varphi}.
		\]
		It follows that
		\begin{equation}\label{eq:green-kernel-oscillatory}
		G_s(\ell,\theta)=0\quad(\ell\geq0),
		\qquad
		G_s(-k,\theta)
		=s^{-k}\frac{\sin(k\varphi)}{\sin\varphi}
		\quad(k\geq1).
		\end{equation}
		
		Thus the kernels in Case~1 are uniformly bounded.\\

\nin		
		\noindent\emph{Case 2: $b(\theta)=2\cosh t$ with
			$0<t<\log s$.}

		The two roots of \eqref{g23} are
		\[
		z_1=se^t,
		\qquad
		z_2=se^{-t}.
		\]
		In this case, since $0<t<\log s$, $|z_1|>1$ and $|z_2|>1$.  As in Case~1, we
		choose the zero solution on $\ell\geq0$.  For $\ell\leq-1$, write
		\[
		G_s(\ell,\theta)=c_1z_1^\ell+c_2z_2^\ell,
		\]
		and then
		\[
		c_1+c_2=0,
		\qquad
		c_1z_1^{-1}+c_2z_2^{-1}=s^{-1}.
		\]
		It follows that
		\[
		c_1=-\frac{1}{2\sinh t},
		\qquad
		c_2=\frac{1}{2\sinh t}.
		\]
		Therefore
		\begin{equation}\label{eq:green-kernel-hyperbolic}
		G_s(\ell,\theta)=0\quad(\ell\geq0),
		\qquad
		G_s(-k,\theta)
		=s^{-k}\frac{\sinh(kt)}{\sinh t}
		\quad(k\geq1).
		\end{equation}
		
		We next obtain a uniform bound.  If
		$0<t\leq\frac12\log s$, then
		\[
		\frac{\sinh(kt)}{\sinh t}
		\leq k e^{(k-1)t},
		\]
		and hence
		\[
		|G_s(-k,\theta)|
		\leq k s^{-k}e^{(k-1)t}
		\leq k s^{-(k+1)/2}
		\leq k s_0^{-(k+1)/2}.
		\]
		If $\frac12\log s\leq t<\log s$, then
		\[
		|G_s(-k,\theta)|
		\leq
		\frac{(e^t/s)^k}{2\sinh t}
		\leq
		\frac{1}{2\sinh(\frac12\log s_0)}.
		\]
		Thus the kernels in Case~2 are uniformly bounded.
		
		\medskip
		\noindent\emph{Case 3: $b(\theta)=2\cosh t$ and $t>\log s$.}
		The two roots are
		\[
		z_1=se^t,
		\qquad
		z_2=se^{-t},
		\]
		so that $0<z_2<1<z_1$.  In this case, we choose the $z_2^\ell$ 
		as $\ell\to+\infty$ and the $z_1^\ell$   as
		$\ell\to-\infty$.  Thus we first write
		\[
		G_s(\ell,\theta)=
		\begin{cases}
		c_+z_2^\ell,&\ell\geq0,\\
		c_-z_1^\ell,&\ell\leq-1.
		\end{cases}
		\]
		We again need two conditions to determine the two constants.  First,
		the equation \eqref{eq:green-kernel} at $\ell=-1$ gives
		\begin{equation}\label{g26}
		s^{-1}c_++s c_-z_1^{-2}-b(\theta)c_-z_1^{-1}=0.
		\end{equation}

		Second, the equation \eqref{eq:green-kernel} at $\ell=0$ gives
		\begin{equation}\label{g27}
		s^{-1}c_+z_2+s c_-z_1^{-1}-b(\theta)c_+=1.
		\end{equation}
		Solving \eqref{g26} and \eqref{g27}, we obtain
		\[
		c_+=c_-=-\frac{1}{2\sinh t}.
		\]
		Consequently,
		\begin{equation}\label{eq:green-kernel-split}
		G_s(\ell,\theta)
		=-\frac{1}{2\sinh t}
		\begin{cases}
		z_2^\ell,&\ell\geq0,\\
		z_1^\ell,&\ell\leq-1.
		\end{cases}
		\end{equation}
		Finally,
		\[
		|G_s(\ell,\theta)|
		\leq\frac{1}{2\sinh t}
		\leq\frac{1}{2\sinh(\log s_0)},
		\]
		where the last inequality follows since $s_0>1$ and hence $\log (s_0)>0.$
		Combining the three cases, we have proved that
		\begin{equation}\label{eq:uniform-green-kernel-bound}
		\sup_{s_0\leq s\leq\rho}
		\sup_{\theta\in\T}
		\sup_{\ell\in\Z}
		|G_s(\ell,\theta)|
		\leq C(s_0,\lambda).
		\end{equation}
The three cases above cover all $\theta\in\T$ except possibly finitely
		many points.  Define $G_s$ to be zero on this exceptional set.  This
		does not affect the construction, since the identities are understood
		almost everywhere in $\theta$.
		
By \eqref{eq:uniform-green-kernel-bound}, the series defining $f_m$ in
		\eqref{g29} converges in $L^2(\T)$, and
		\begin{align}
		\|f_m\|_{L^2(\T)}
		&\leq
		\sum_{\ell\in\Z}
		\|G_s(m-\ell,\cdot)h_\ell\|_{L^2(\T)}\nonumber\\
		&\leq C(s_0,\lambda)
		\sum_{\ell\in\Z}\|h_\ell\|_{L^2(\T)}.\label{g15}
		\end{align}
		
		Equation~\eqref{eq:green-kernel} and absolute convergence in
		$L^2(\T)$ give, for every $m\in\Z$,
		\[
		\begin{aligned}
		(L_sf)_m
		&=\sum_{\ell\in\Z}
		\bigl[s^{-1}G_s(m-\ell+1,\cdot)+sG_s(m-\ell-1,\cdot)\\
		&\hspace{38mm}-b(\cdot)G_s(m-\ell,\cdot)\bigr]h_\ell\\
		&=\sum_{\ell\in\Z}\delta_0(m-\ell)h_\ell=h_m.
		\end{aligned}
		\]
		By Lemma~\ref{lem:uniqueness-X}, $f=(f_m)$ is the unique solution of
		$L_sf=h$.	 Therefore \eqref{g15} implies \eqref{eq:uniform-inverse-estimate}.
		It remains to prove \eqref{eq:one-sided-inverse-estimate}.	By \eqref{eq:critical-weight-identity},
		\[
		\max_\theta b(\theta)=\lambda+2
		=\rho+\rho^{-1}=2\cosh(\log\rho).
		\]
		At $s=\rho$, only the kernels in Cases~1 and~2 are
		therefore needed.  Both satisfy
		\[
		G_\rho(\ell,\theta)=0\qquad(\ell\geq0).
		\]
		Taking $s=\rho$ in the formula for $f$, we obtain
		\begin{equation}\label{g17}
		f_m(\theta)
		=\sum_{\ell>m}G_\rho(m-\ell,\theta)h_\ell(\theta).
		\end{equation}
		Now \eqref{eq:one-sided-inverse-estimate} follows from
		\eqref{eq:uniform-green-kernel-bound} and \eqref{g17}.
		
	\end{proof}
	
	Given the above preparation, we can now finish the proof of Theorem~\ref{thm1:absence}.\\
	
	\nin
	\textbf{Proof of Theorem~\ref{thm1:absence}.}
	By \eqref{eq:checkerboard-symmetry}, it suffices to consider
	$\lambda\in(0,4)$. 
	For $1<s<\rho$, define
	\[
	f_m^{(s)}=s^m\widehat U_m.
	\]
	Thus $\|f_m^{(s)}\|_{L^2(\T)}= s^m\norm{U_m}_{\cH}.$
	As indicated already, we will obtain uniform (in $1<s <\rho$) control on the norms of $f_m^{(s)}$ that will allow us to conclude the same about $s=\rho.$

\nin	
The exponential decay result from Proposition~\ref{thm:subcritical-exponential-decay} implies
	that $f^{(s)}\in\ell^1(\Z;L^2(\T))$.  
	The recurrence \eqref{eq:noncritical-slice-recurrence} gives the exact
	identity
	\begin{equation}\label{eq:weighted-source-identity}
	(L_sf^{(s)})_m=s^m\widehat H_m.
	\end{equation}
	where $H_{m}=V_mU_m.$
	Let us now set
	\[
	Y_s=\sup_{m\in\Z}s^m\norm{U_m}_{\cH}.
	\]
	Proposition~\ref{thm:uniform-inverse} (in particular \eqref{eq:uniform-inverse-estimate}) and the pointwise
	decay of $V$ yield
	\begin{equation}\label{eq:critical-rate-estimate}
	Y_s
	\leq C\sum_m\la m\ra^{-\kappa}s^m\norm{U_m}_{\cH}.
	\end{equation}
	Split the sum into $|m|\leq R$ and $|m|>R$.  The finite part is bounded
	uniformly for $s\leq\rho$, while the tail is at most
	\[
	C Y_s\sum_{|m|>R}\la m\ra^{-\kappa}.
	\]
	Since $\kappa>1$, we may choose $R$ so large that  $C\sum_{|m|>R}\la m\ra^{-\kappa}$  is less than $1/2$.  
		We obtain
	\[
	\sup_{s_0\leq s<\rho}Y_s\leq C,
	\]
	where $C$ is independent of $s$. 
	This crucially allows us to now extend the above estimate to $s=\rho.$ Letting $s\uparrow\rho$ gives the endpoint bound
	\begin{equation}\label{eq:critical-rate-bound}
	Y_\rho:=\sup_{m\in\Z}\rho^m\norm{U_m}_{\cH}<\infty.
	\end{equation}
	Set
	\[
	F_m=\rho^m\widehat U_m.
	\]
	By \eqref{eq:critical-rate-bound},
	\begin{equation}\label{eq:endpoint-source-summability}
	\sum_m\norm{(L_\rho F)_m}_{L^2(\T)}
	=\sum_m\rho^m\norm{H_m}_{\cH}
	\leq C Y_\rho\sum_m\la m\ra^{-\kappa}<\infty.
	\end{equation}
	Apply the one-sided estimate \eqref{eq:one-sided-inverse-estimate}, and set
	\[
	T_M=\sup_{m\geq M}\norm{F_m}_{L^2(\T)}.
	\]
	Then
	\begin{align}
	T_M
	&\leq C\sum_{\ell\geq M+1}\rho^\ell\norm{H_\ell}_{\cH}\\
	&\leq C T_M\sum_{\ell\geq M+1}\la \ell\ra^{-\kappa}.
	\label{eq:endpoint-tail-contraction}
	\end{align}
	For sufficiently large $M$, the final coefficient
	$C\sum_{\ell\geq M+1}\la \ell\ra^{-\kappa}$ is strictly smaller
	than one.  Therefore $T_M=0$, so
	\[
	U_m=0\qquad(m\geq M).
	\]
	In particular,  both $U_M$ and $U_{M+1}$ vanish.
	Equation \eqref{eq:noncritical-slice-recurrence} then implies that
	$u=0$ on $\Z^2$.
	\qed
	\bibliographystyle{amsplain}
	\bibliography{intro}

@article{KozmaSchreiber2004,
  author  = {Gady Kozma and Ehud Schreiber},
  title   = {An asymptotic expansion for the discrete harmonic potential},
  journal = {Electronic Journal of Probability},
  volume  = {9},
  number  = {1},
  pages   = {1--17},
  year    = {2004},
  doi     = {10.1214/EJP.v9-170}
}

@article {AndoIsozakiMorioka2016,
    AUTHOR = {Ando, Kazunori and Isozaki, Hiroshi and Morioka, Hisashi},
     TITLE = {Spectral properties of {S}chr\"odinger operators on perturbed
              lattices},
   JOURNAL = {Ann. Henri Poincar\'e},
  FJOURNAL = {Annales Henri Poincar\'e},
    VOLUME = {17},
      YEAR = {2016},
    NUMBER = {8},
     PAGES = {2103--2171},
  MRNUMBER = {3522026},
       DOI = {10.1007/s00023-015-0430-0},
       URL = {https://doi.org/10.1007/s00023-015-0430-0},
}

@article {AndoIsozakiMorioka2025,
    AUTHOR = {Ando, Kazunori and Isozaki, Hiroshi and Morioka, Hisashi},
     TITLE = {A remark on the absence of eigenvalues in continuous spectra
              for discrete {S}chr\"odinger operators on periodic lattices},
  journal={arXiv preprint arXiv:2411.03577},
  year={2025}
}

@article {BouRabeeCoopermanGanguly2025,
    AUTHOR = {Bou-Rabee, Ahmed and Cooperman, William and Ganguly, Shirshendu},
     TITLE = {Unique continuation on planar graphs},
   JOURNAL = {Discrete Anal.},
  FJOURNAL = {Discrete Analysis},
      YEAR = {2025},
     PAGES = {Paper No. 16, 12},
       DOI = {10.19086/da.144015},
       URL = {https://doi.org/10.19086/da.144015},
}

@article {BourgainKenig2005,
    AUTHOR = {Bourgain, Jean and Kenig, Carlos E.},
     TITLE = {On localization in the continuous {A}nderson--{B}ernoulli model
              in higher dimension},
   JOURNAL = {Invent. Math.},
  FJOURNAL = {Inventiones Mathematicae},
    VOLUME = {161},
      YEAR = {2005},
    NUMBER = {2},
     PAGES = {389--426},
       DOI = {10.1007/s00222-004-0435-7},
       URL = {https://doi.org/10.1007/s00222-004-0435-7},
}

@article {BuhovskyLogunovMalinnikovaSodin2022,
    AUTHOR = {Buhovsky, Lev and Logunov, Alexander and Malinnikova, Eugenia and
              Sodin, Mikhail},
     TITLE = {A discrete harmonic function bounded on a large portion of
              {$\mathbb Z^2$} is constant},
   JOURNAL = {Duke Math. J.},
  FJOURNAL = {Duke Mathematical Journal},
    VOLUME = {171},
      YEAR = {2022},
    NUMBER = {6},
     PAGES = {1349--1378},
       URL = {https://arxiv.org/abs/1712.07902},
}

@article {ChuLiuLyu2026,
    AUTHOR = {Chu, Jifeng and Liu, Wencai and Lyu, Kang},
     TITLE = {Optimal bound for the embedded eigenvalue of one-dimensional
              discrete {S}chr\"odinger operators with decaying potentials},
      journal={arXiv preprint arXiv:2609.03224},
  year={2026}
}

@article {DingSmart2020,
    AUTHOR = {Ding, Jian and Smart, Charles K.},
     TITLE = {Localization near the edge for the {A}nderson--{B}ernoulli model
              on the two-dimensional lattice},
   JOURNAL = {Invent. Math.},
  FJOURNAL = {Inventiones Mathematicae},
    VOLUME = {219},
      YEAR = {2020},
    NUMBER = {2},
     PAGES = {467--506},
       URL = {https://arxiv.org/abs/1809.09041},
}

@article {LiZhang2022,
    AUTHOR = {Li, Linjun and Zhang, Lingfu},
     TITLE = {{A}nderson--{B}ernoulli localization on the three-dimensional
              lattice and discrete unique continuation principle},
   JOURNAL = {Duke Math. J.},
  FJOURNAL = {Duke Mathematical Journal},
    VOLUME = {171},
      YEAR = {2022},
    NUMBER = {2},
     PAGES = {327--415},
       URL = {https://arxiv.org/abs/1906.04350},
}

@article {FrankSimon2017,
    AUTHOR = {Frank, Rupert L. and Simon, Barry},
     TITLE = {Eigenvalue bounds for {S}chr\"odinger operators with complex
              potentials. {II}},
   JOURNAL = {J. Spectr. Theory},
  FJOURNAL = {Journal of Spectral Theory},
    VOLUME = {7},
      YEAR = {2017},
    NUMBER = {3},
     PAGES = {633--658},
      ISSN = {1664-039X},
   MRCLASS = {35P15 (35J10 81Q12)},
  MRNUMBER = {3713021},
       DOI = {10.4171/JST/173},
       URL = {https://doi.org/10.4171/JST/173},
}

@article {HundertmarkJexLange2023,
    AUTHOR = {Hundertmark, Dirk and Jex, Michal and Lange, Markus},
     TITLE = {Quantum systems at the brink: existence of bound states,
              critical potentials, and dimensionality},
   JOURNAL = {Forum Math. Sigma},
  FJOURNAL = {Forum of Mathematics. Sigma},
    VOLUME = {11},
      YEAR = {2023},
     PAGES = {Paper No. e61, 29},
  MRNUMBER = {4616434},
       DOI = {10.1017/fms.2023.39},
       URL = {https://doi.org/10.1017/fms.2023.39},
}

@article {IonescuJerison2003,
    AUTHOR = {Ionescu, Alexandru D. and Jerison, David},
     TITLE = {On the absence of positive eigenvalues of {S}chr\"odinger
              operators with rough potentials},
   JOURNAL = {Geom. Funct. Anal.},
  FJOURNAL = {Geometric and Functional Analysis},
    VOLUME = {13},
      YEAR = {2003},
    NUMBER = {5},
     PAGES = {1029--1081},
  MRNUMBER = {2024415},
       DOI = {10.1007/s00039-003-0439-2},
       URL = {https://doi.org/10.1007/s00039-003-0439-2},
}

@article {IsozakiKorotyaev2012,
    AUTHOR = {Isozaki, Hiroshi and Korotyaev, Evgeny},
     TITLE = {Inverse problems, trace formulae for discrete {S}chr\"odinger
              operators},
   JOURNAL = {Ann. Henri Poincar\'e},
  FJOURNAL = {Annales Henri Poincar\'e},
    VOLUME = {13},
      YEAR = {2012},
    NUMBER = {4},
     PAGES = {751--788},
  MRNUMBER = {2913620},
       DOI = {10.1007/s00023-011-0141-0},
       URL = {https://doi.org/10.1007/s00023-011-0141-0},
}

@article {IsozakiMorioka2014,
    AUTHOR = {Isozaki, Hiroshi and Morioka, Hisashi},
     TITLE = {A {R}ellich type theorem for discrete {S}chr\"odinger
              operators},
   JOURNAL = {Inverse Probl. Imaging},
  FJOURNAL = {Inverse Problems and Imaging},
    VOLUME = {8},
      YEAR = {2014},
    NUMBER = {2},
     PAGES = {475--489},
       DOI = {10.3934/ipi.2014.8.475},
       URL = {https://doi.org/10.3934/ipi.2014.8.475},
}

@article {JexStampach2025,
    AUTHOR = {Jex, Michal and \v{S}tampach, Franti\v{s}ek},
     TITLE = {On the ground state of lattice {S}chr\"odinger operators},
   JOURNAL = {J. Spectr. Theory},
  FJOURNAL = {Journal of Spectral Theory},
    VOLUME = {15},
      YEAR = {2025},
    NUMBER = {2},
     PAGES = {647--678},
       DOI = {10.4171/JST/558},
       URL = {https://doi.org/10.4171/JST/558},
}

@article {JitomirskayaLiu2019,
    AUTHOR = {Jitomirskaya, Svetlana and Liu, Wencai},
     TITLE = {Noncompact complete {R}iemannian manifolds with dense
              eigenvalues embedded in the essential spectrum of the
              {L}aplacian},
   JOURNAL = {Geom. Funct. Anal.},
  FJOURNAL = {Geometric and Functional Analysis},
    VOLUME = {29},
      YEAR = {2019},
    NUMBER = {1},
     PAGES = {238--257},
       DOI = {10.1007/s00039-019-00480-w},
       URL = {https://doi.org/10.1007/s00039-019-00480-w},
}

@article {Kato1959,
    AUTHOR = {Kato, Tosio},
     TITLE = {Growth properties of solutions of the reduced wave equation
              with a variable coefficient},
   JOURNAL = {Comm. Pure Appl. Math.},
  FJOURNAL = {Communications on Pure and Applied Mathematics},
    VOLUME = {12},
      YEAR = {1959},
     PAGES = {403--425},
  MRNUMBER = {0108633},
       DOI = {10.1002/cpa.3160120302},
       URL = {https://doi.org/10.1002/cpa.3160120302},
}

@article {KochTataru2006,
    AUTHOR = {Koch, Herbert and Tataru, Daniel},
     TITLE = {Carleman estimates and absence of embedded eigenvalues},
   JOURNAL = {Comm. Math. Phys.},
  FJOURNAL = {Communications in Mathematical Physics},
    VOLUME = {267},
      YEAR = {2006},
    NUMBER = {2},
     PAGES = {419--449},
  MRNUMBER = {2252331},
       DOI = {10.1007/s00220-006-0060-y},
       URL = {https://doi.org/10.1007/s00220-006-0060-y},
}

@incollection {Kuchment1991,
    AUTHOR = {Kuchment, Peter A.},
     TITLE = {On the {F}loquet theory of periodic difference equations},
 BOOKTITLE = {Geometrical and algebraical aspects in several complex variables
              ({C}etraro, 1989)},
    SERIES = {Sem. Conf.},
    VOLUME = {8},
     PAGES = {201--209},
 PUBLISHER = {EditEl},
   ADDRESS = {Rende},
      YEAR = {1991},
  MRNUMBER = {1222215},
}

@article {KuchmentVainberg2000,
    AUTHOR = {Kuchment, Peter and Vainberg, Boris},
     TITLE = {On absence of embedded eigenvalues for {S}chr\"odinger
              operators with perturbed periodic potentials},
   JOURNAL = {Comm. Partial Differential Equations},
  FJOURNAL = {Communications in Partial Differential Equations},
    VOLUME = {25},
      YEAR = {2000},
    NUMBER = {9--10},
     PAGES = {1809--1826},
       DOI = {10.1080/03605300008821568},
       URL = {https://doi.org/10.1080/03605300008821568},
}

@article {LiuCriteria2021,
    AUTHOR = {Liu, Wencai},
     TITLE = {Criteria for embedded eigenvalues for discrete
              {S}chr\"odinger operators},
   JOURNAL = {Int. Math. Res. Not. IMRN},
  FJOURNAL = {International Mathematics Research Notices. IMRN},
      YEAR = {2021},
    NUMBER = {20},
     PAGES = {15803--15832},
  MRNUMBER = {4329883},
       DOI = {10.1093/imrn/rnz262},
       URL = {https://doi.org/10.1093/imrn/rnz262},
}

@article {LiuIrreducibility2022,
    AUTHOR = {Liu, Wencai},
     TITLE = {Irreducibility of the {F}ermi variety for discrete periodic
              {S}chr\"odinger operators and embedded eigenvalues},
   JOURNAL = {Geom. Funct. Anal.},
  FJOURNAL = {Geometric and Functional Analysis},
    VOLUME = {32},
      YEAR = {2022},
    NUMBER = {1},
     PAGES = {1--30},
  MRNUMBER = {4388759},
       DOI = {10.1007/s00039-021-00587-z},
       URL = {https://doi.org/10.1007/s00039-021-00587-z},
}

@article {LiuMatosTreuer2025,
    AUTHOR = {Liu, Wencai and Matos, Rodrigo and Treuer, John N.},
     TITLE = {Sharp decay rate for eigenfunctions of perturbed periodic
              {S}chr\"odinger operators},
   JOURNAL = {Nonlinearity},
  FJOURNAL = {Nonlinearity},
    VOLUME = {38},
      YEAR = {2025},
    NUMBER = {4},
     PAGES = {Paper No. 045028, 25},
       DOI = {10.1088/1361-6544/adc653},
       URL = {https://doi.org/10.1088/1361-6544/adc653},
}

@article {LiuOng2020,
    AUTHOR = {Liu, Wencai and Ong, Darren C.},
     TITLE = {Sharp spectral transition for eigenvalues embedded into the
              spectral bands of perturbed periodic operators},
   JOURNAL = {J. Anal. Math.},
  FJOURNAL = {Journal d'Analyse Math\'ematique},
    VOLUME = {141},
      YEAR = {2020},
    NUMBER = {2},
     PAGES = {625--661},
       DOI = {10.1007/s11854-020-0111-x},
       URL = {https://doi.org/10.1007/s11854-020-0111-x},
}

@article {Shipman2014,
    AUTHOR = {Shipman, Stephen P.},
     TITLE = {Eigenfunctions of unbounded support for embedded eigenvalues of
              locally perturbed periodic graph operators},
   JOURNAL = {Comm. Math. Phys.},
  FJOURNAL = {Communications in Mathematical Physics},
    VOLUME = {332},
      YEAR = {2014},
    NUMBER = {2},
     PAGES = {605--626},
  MRNUMBER = {3257657},
       DOI = {10.1007/s00220-014-2113-y},
       URL = {https://doi.org/10.1007/s00220-014-2113-y},
}

@article {TadanoTaira2019,
    AUTHOR = {Tadano, Yukihide and Taira, Kouichi},
     TITLE = {Uniform bounds of discrete {B}irman--{S}chwinger operators},
   JOURNAL = {Trans. Amer. Math. Soc.},
  FJOURNAL = {Transactions of the American Mathematical Society},
    VOLUME = {372},
      YEAR = {2019},
    NUMBER = {7},
     PAGES = {5243--5262},
      ISSN = {0002-9947},
  MRNUMBER = {4009459},
       DOI = {10.1090/tran/7882},
       URL = {https://doi.org/10.1090/tran/7882},
}

@article {WignerVonNeumann1929,
    AUTHOR = {von Neumann, John and Wigner, Eugene P.},
     TITLE = {\"Uber merkw\"urdige diskrete {E}igenwerte},
   JOURNAL = {Phys. Z.},
  FJOURNAL = {Physikalische Zeitschrift},
    VOLUME = {30},
      YEAR = {1929},
     PAGES = {465--467},
}

@misc {WikipediaChebyshev,
    AUTHOR = {{Wikipedia contributors}},
     TITLE = {{Chebyshev} polynomials},
 HOWPUBLISHED = {\url{https://en.wikipedia.org/wiki/Chebyshev_polynomials}},
      NOTE = {Accessed September 4, 2026},
}

@book {Zworski2012,
    AUTHOR = {Zworski, Maciej},
     TITLE = {Semiclassical Analysis},
    SERIES = {Graduate Studies in Mathematics},
    VOLUME = {138},
 PUBLISHER = {American Mathematical Society},
   ADDRESS = {Providence, RI},
      YEAR = {2012},
}
\end{document}